%% file: orthoQLSPaperArxiv.tex
\documentclass[]{eptcs}
\usepackage{breakurl}
\usepackage{docmute}
\usepackage{hyperref}
\usepackage[sort,numbers]{natbib}
\usepackage{multirow}
\usepackage{esvect}
\usepackage{amsmath,amsfonts,amssymb,xspace}
\usepackage{amsthm}
\usepackage{tikz}
\usetikzlibrary{fit,shapes}
\usepackage{etoolbox}
\usepackage{datetime}
\usepackage{amssymb}
\newenvironment{pic}[1][]
{\begin{aligned}\begin{tikzpicture}[#1]}
{\end{tikzpicture}\end{aligned}}
\newcommand{\edges}[1][]%
{
}
\theoremstyle{plain}
\newtheorem{theorem}{Theorem}
\newtheorem{lemma}[theorem]{Lemma}

\newtheorem{corollary}[theorem]{Corollary}

\theoremstyle{definition} 
\newtheorem{definition}[theorem]{Definition} 
\newtheorem{example}[theorem]{Example} 
 
\newtheorem*{conj*}{Conjecture}

\newtheorem{remark}{Remark}
\newcommand{\tr}{\text{tr}}

\makeatletter
\def\calign@preamble{%
   &\hfil\strut@
    \setboxz@h{\@lign$\m@th\displaystyle{##}$}%
    \ifmeasuring@\savefieldlength@\fi
    \set@field
    \hfil
    \tabskip\alignsep@
}
\let\cmeasure@\measure@
\patchcmd\cmeasure@{\divide\@tempcntb\tw@}{}{}{}
\patchcmd\cmeasure@{\divide\@tempcntb\tw@}{}{}{}
\patchcmd\cmeasure@{\ifodd\maxfields@
  \global\advance\maxfields@\@ne
  \fi}{}{}{}    
  
\makeatother

\newcommand\tinymatrix[1]
{\left( \hspace{-2pt} \renewcommand\thickspace{\kern2pt} \scriptstyle\begin{smallmatrix} #1 \end{smallmatrix} \hspace{-2pt} \right)}

\newcommand\ignore[1]{}

\newcommand\cat[1]{\ensuremath{\mathbf{#1}}}
\renewcommand\dag{\ensuremath{\dagger}}

\newcommand\id[1][]{\ensuremath{\mathrm{id}_{#1}}}

\newcommand\C{\ensuremath{\mathbb{C}}}
\newcommand{\inprod}[2]{\ensuremath{\langle #1\hspace{0.5pt}|\hspace{0.5pt}#2 \rangle}}

\newcommand\pdag{{\phantom{\dagger}}}

\usepackage{color}

\def\arraystretch{1.0}
\newcommand\grid[1]{\ensuremath{\def\arraystretch{1.4}\begin{array}{|c|c|c|c|c|c|c|c|c|}\hline#1\\\hline\end{array}}}

\newcommand\super[2]{\stackrel{\makebox[0pt]{\smash{\tiny #1}}}{#2}}

\newcommand\bra[1]{\langle #1|}
\newcommand\ket[1]{{|} #1 \rangle}
\newcommand\braket[2]{\langle #1 | #2 \rangle}
\newcommand\ketbra[2]{|#1 \rangle \hspace{-1pt} \langle #2 |}

\newcounter{jamiecomment}
\usepackage{amsbsy}
\usepackage[normalem]{ulem}
\allowdisplaybreaks[1]
\input{book-header}

\def\titlerunning{Constructing Mutually Unbiased Bases from Quantum Latin Squares}
\def\authorrunning{Benjamin Musto}
\begin{document}

\title{Constructing Mutually Unbiased Bases \\ from Quantum Latin Squares}
\author{Benjamin Musto
\institute{Department of Computer Science\\ University of Oxford}
\email{benjamin.musto@cs.ox.ac.uk}
}
\date{\today}
\maketitle

\begin{abstract}
We introduce \textit{orthogonal quantum Latin squares}, which restrict to traditional orthogonal Latin squares, and investigate their application in quantum information science. We use quantum Latin squares  to build  maximally entangled bases, and  show how mutually unbiased maximally entangled bases can be constructed  in square dimension from orthogonal quantum Latin squares. We also compare our construction to an existing construction due to Beth and Wocjan~\cite{ortho1} and show that ours is strictly more general. \end{abstract}
\section{Introduction}
In this paper we introduce a notion of orthogonality between \textit{quantum Latin squares} (QLSs)~\cite{mypaper1}, mathematical objects which generalise \textit{Latin squares}.  We use this concept to give a new construction of \textit{maximally entangled mutually unbiased bases} (MUBs), extending existing known techniques for Latin squares~\cite{ortho1,ortho2}. In addition we prove that our construction can produce bases that are unobtainable by existing methods~\cite{ortho1,ortho2}.
We also introduce the concept of \textit{mutually weak orthogonal quantum Latin squares }(MOQLS) which generalise \textit{mutually orthogonal Latin squares} (MOLS), about which a significant body of research exists in connection with quantum information, and particularly pertaining to the connection between MOLS and MUBs~\cite{klapp2,paterek,cao2016more}.
Mutually unbiased bases are of fundamental importance to quantum information, as they capture the physical notion of complementary observables, quantities that cannot be simultaneously measured. Entanglement is one of the central phenomena of quantum theory that is at the foundation  of quantum information and computation. 

The results presented in this paper were developed using the graphical calculus of categorical quantum mechanics (CQM), and we have made use of it where we believe it elucidates some detail. For those unfamiliar with  CQM, there is a short introduction of the concepts necessary to understand this paper in Appendix~\ref{apx:cqm}; for a thorough introduction please refer to~\cite{qcs, abramskycoecke2004, surveycategoricalquantummechanics}. Everything that we present here is in the category $\cat{FHilb}$ of finite Hilbert spaces and linear maps, but could be interpreted in any monoidal category such as $\cat{Rel}$ with \textit{quantum-like} properties, which have been extensively researched as quantum toy theories.

We start with a definition of  quantum Latin squares.
\begin{definition}
\label{def:qls}
A \textit{quantum Latin square of order $n$} is an  $n \times n$ array of elements of the Hilbert space $\C^n$, such that every row and every column is an orthonormal basis. 
\end{definition}
\begin{example}\label{ex:qls}
Here is an example of a quantum Latin square given in terms of the computational basis states $\ket i$ for $i\in \{0,...,9\}$, and the following states:

\begin{minipage}[t]{6.2cm}
\begin{align}\label{eq:a}
\ket a&:=\frac{1}{\sqrt{3}}(\ket 3+ \ket 4 + i\ket 5)\\ \label{eq:b}
\ket b&:=\frac{1}{\sqrt{6}}(2\ket 3- \ket 4 + i\ket 5)\\  \label{eq:c}
\ket c&:=\frac{1}{\sqrt{14}}(-2i\ket 3-i\ket 4 + 3\ket 5)
\end{align}
\end{minipage}
\begin{minipage}[t]{7cm}
\begin{align}\label{eq:abg}
\ket \alpha &:=\frac{1}{\sqrt{3}}(\ket 0+ \ket 1 + \ket 2)
\\ \label{eq:abg2}
\ket \beta &:=\frac{1}{\sqrt{3}}(\ket 0 +e^{\frac{2 \pi i}{3}} \ket 1 + e^{\frac{-2 \pi i}{3}}\ket 2 )\\ \label{eq:abg3}
\ket \gamma &:=\frac{1}{\sqrt{3}}(\ket 0 +e^{\frac{-2 \pi i}{3}} \ket 1 + e^{\frac{2 \pi i}{3}}\ket 2 )
\end{align}
\end{minipage}
\hspace{-67pt}\begin{equation*}\grid{
 \ket{0}  &  \ket{2}  &  \ket{1}  &  \ket{3}  &  \ket{5}  &  \ket{4}  &  \ket {6}  &  \ket {8}  &  \ket 7 
\\ \hline
 \ket 2 
&  \ket{1} 
&  \ket{0} 
&  \ket 5  &  \ket 4  &  \ket 3  &  \ket 8  &  \ket 7  &  \ket 6 
\\ \hline
 \ket{1} 
&  \ket{0} 
&  \ket{2} 
&  \ket{4}  &  \ket{3}  &  \ket 5  &  \ket 7  &  \ket 6  &  \ket 8  
\\ \hline
 \ket{6}  &  \ket{8}  &  \ket{7}  &  \ket{0}  &  \ket 2  &  \ket 1  &  \ket 3  &  \ket 5  &  \ket 4 
\\ \hline
\ket 8 & \ket 7 & \ket 6 & \ket 2 & \ket 1 & \ket 0 & \ket 5 & \ket 4 & \ket 3 \\ \hline

\ket 7 & \ket 6 & \ket 8 & \ket 1 & \ket 0 & \ket 2 & \ket 4 & \ket 3 & \ket 5  \\ \hline

 \ket  a & \ket c & \ket b & \ket 6 & \ket 8 & \ket 7 &  \ket \alpha  & \ket \gamma & \ket \beta  \\ \hline

\ket c & \ket b&\ket  a &\ket 8&\ket 7&\ket 6&\ket \gamma&\ket \beta&\ket \alpha \\ \hline

\ket b & \ket  a&\ket c &\ket 7&\ket 6&\ket 8&\ket \beta&\ket \alpha&\ket \gamma}
\end{equation*}
\end{example}
\noindent It can be checked that every row and every column is an orthonormal basis.

\begin{definition}[Latin square] A \textit{Latin square} is  a QLS with entries that all come from the computational basis. For those who are familiar with the traditional definition, it is recovered by mapping each computational basis state to a different symbol.\end{definition}
The main result of this paper is a construction of mutually unbiased maximally entangled bases from orthogonal QLSs. We now define the necessary concepts. 
\begin{definition}[Mutually unbiased bases]
Two orthonormal bases $\ket {a_i}$ and $\ket {b_j}$ for a Hilbert space $\mathcal H $ of dimension $n$ are \textit{mutually unbiased} when, for all $i,j \,\in \, \, \{0 ,...,n-1\}$~\cite{bandyopadhyay}:
\begin{equation}\label{def:mub}
|\inprod{a_i}{b_j}|^2 \, = \, \frac{1}{n}
\end{equation}
\end{definition}

\begin{definition}[Maximally entangled state]
A \textit{maximally entangled state} of a bipartite system is a  state $\ket \psi$ of a product Hilbert space $\mathcal H_A \otimes \mathcal H_B$ with dim($\mathcal H_B)=n,$ such that the partial trace over one of the systems of its density operator $\rho_{AB}=\ketbra{\psi}{\psi}$ is proportional to the identity. i.e~\cite{medefinition}.
\begin{equation}
\rho_A :=\sum_{k=0}(\id[A] \otimes \bra k) \rho_{AB} (\id[A] \otimes \ket{k})=\frac{1}{n}\id[A \otimes B]
\end{equation} 
\end{definition}

\begin{remark}For the Hilbert space $\mathcal H \otimes \mathcal H$ with dim$(\mathcal H)=n$, all maximally entangled states are of the following form, where $U$  is a unitary linear map and $\tinyspider$  is the classical structure (see Appendix~\ref{apx:cqm}) corresponding to the orthonormal basis $\ket k$~\cite{vedral}:
\begin{equation}\label{mes}\ket U :=\frac{1}{\sqrt{n}}\sum_{k=0}^{n-1}\ket k \otimes U\ket k \qquad \text{ or equivalently }\qquad \ket U :=\frac{1}{\sqrt{n}}\begin{pic}
\node[blackdot] (b) {};
\node (1) at (-0.5,1) {};
\node (2) at (0.5,1) {};
\node[morphism,scale=0.7] (U) at (0.5,0.5) {$U$};
\draw[string,in=left,out=down,looseness=1.3] (1.center) to (b.center);
\draw[string,in=up,out=down,looseness=1.3] (2.center) to (U.north);
\draw[string,in=down,out=right,looseness=1.1] (b.center) to (U.south);
\end{pic}\end{equation}
\end{remark}
\begin{definition}[Maximally entangled basis]\label{def:meb}
A \textit{maximally entangled basis} (MEB) for a bipartite system represented by a  tensor product Hilbert space $\mathcal H \otimes \mathcal H$, is an orthonormal basis such that each basis state is maximally entangled. \end{definition}
Two MEBs
$\mathcal A:=\ket{A_i}$ and $\mathcal B:=\ket{B_i}$ are equivalent when there exist unitaries $U$ and $V$ and complex numbers of modulus $1$, $c_i$ such that:
\begin{equation}\label{eq:eqimeb}
   \begin{pic}
      \node[state,xscale=1.5] {$A_i$};
      \node(1) at (-0.5,1) {};  
      \node(2) at (0.5,1) {}; 
      \draw (1.center) to +(0,-1);
      \draw (2.center) to +(0,-1);
   \end{pic}
   \quad = \quad
   c_i
   \begin{pic}
      \node[state,xscale=1.5] {$B_i$};
      \node(1)[morphism,scale=0.7] at (-0.5,0.5){$U$};
      \node(2)[morphism,scale=0.7] at (0.5,0.5){$V$};
      \node at (-0.5,1) {};  
      \node at (0.5,1) {};
      \draw (-0.5,0) to (1.south); 
      \draw (0.5,0) to (2.south);
      \draw (-0.5,1) to (1.north);
      \draw (0.5,1) to (2.north);
   \end{pic}
\end{equation}

 In Section~\ref{sec:main} we introduce our main result, the most general construction of mutually unbiased bases of the three presented in this paper. We introduce orthogonal quantum Latin squares and show how they can be used to construct MUBs, and we construct an explicit example. In Section~\ref{sec:left} we start with traditional orthogonality of Latin squares and then show that the definition of orthogonality for QLSs in Section~\ref{sec:main} generalises it. In Section~\ref{sec:bw} we present Beth and Wocjan's construction for MUBs in square dimension, and show that ours is strictly more general. In Section~\ref{sec:dual} we explain the correspondence between unitary error bases and maximally entangled bases and introduce mutually unbiased error bases. Finally in Section~\ref{sec:moqls} we introduce mutually weak  orthogonal quantum Latin squares, which generalise mutually orthogonal Latin squares.
\subsection*{Acknowledgements} The author is grateful to Dominic Verdon and Jamie Vicary for useful discussions, and to EPSRC for financial support.

\section{New construction for square dimension MUBs}\label{sec:main}
In this section we introduce the main result of this paper, a new construction for mutually unbiased maximally entangled bases. In order to formulate our construction we introduce \textit{weak orthogonal quantum Latin squares} which, as we will show in Section~\ref{sec:left},  reduce to  traditional orthogonal  Latin squares. It will be useful to introduce some notation for quantum Latin squares. Given a QLS $\mathcal Q$ we will denote the vector appearing in the $i^{th}$ column of the $j^{th}$ row as $\ket{Q_{ij}}$.

 Before the main result it will be requisite to define generalised Hadamards.
\begin{definition}[Hadamard, see \cite{hadamard}, Definition~2.1]
A \textit{Hadamard matrix of order $n$} is an $n \times n$ matrix $H$ with the following properties for all $i,j$, which we write in both matrix and index form:
\begin{align}
\label{had1}
|H_{ij}| &=1
&
H_{ij} ^\pdag H_{ij} ^* &= 1
\\
\label{had2}
H \circ H^\dag &= n \,\mathbb{I}_n
& \textstyle \sum_p H_{ip} ^\pdag H^* _{jp} &= n \, \delta _{ij}
\\
\label{had3}
H ^\dag \circ H &= n \,\mathbb{I}_n
& \textstyle \sum_p H ^*_{pi} H _{pj} ^\pdag &= n \, \delta _{ij}
\end{align}
\end{definition} 
We now introduce a method for constructing MEBs given as input a family of Hadamards and a quantum Latin square.  This construction is in fact dual to the quantum shift-and-multiply method for constructing unitary error bases~\cite{mypaper1}, as we will explain  in Section~\ref{sec:dual}.
\begin{definition}[Quantum Latin square maximally entangled basis]\label{def:qlsmeb}
Given a quantum Latin square $\mathcal Q$ and a family of Hadamards $H_j$, a \textit{quantum Latin square maximally entangled basis} $B(\mathcal Q, H_j)$ is defined as follows:
\begin{equation}\label{eq:qlsmeb}
\mathcal A := \left\{ A_{ij} = \frac{1}{\sqrt{n}}\sum_{k=0}^{n-1}\ket k \otimes \ketbra{Q_{kj}}{k}H_j\ket i \text{ such that } i,j \in \{ 0,..,n-1\} \right\}
\end{equation}
\end{definition}
\begin{lemma}\label{lem:qlsmeb}
Quantum Latin square maximally entangled bases are maximally entangled bases. 
\end{lemma}
\begin{proof}
This MEB construction is the dual of the quantum shift-and-multiply basis construction, for a proof of the correctness of that construction see~\cite[Theorem 19]{mypaper1}.
\end{proof}\begin{definition}[Weak orthogonal quantum Latin squares]\label{def:woqls} Given a pair of QLSs $\mathcal P$ and $\mathcal Q$ with vector entries $\ket{P_{ij}}$ and $\ket{Q_{ij}}$ respectively, they are \textit{weak orthogonal} when for all $i,j \in \{0,..,n-1\}$, there exists unique $t \in \{0,...,n-1 \}$ such that:
\begin{equation}\label{eq:woqls}
\sum_{k=0}^{n-1}\ket k \braket{Q_{ki}}{P_{kj}}=\ket t
\end{equation} 
In words: if we take any row from $\mathcal P$ and any row from $\mathcal Q$ and compute the componentwise inner product of their vector entries, the resulting $n$ numbers will always be $n-1$ zeros and a single $1$. If the $1$ appears in the $t^{th}$ column then the output state  of the linear map above will be $\ket t.$
\end{definition} 
\begin{example}\label{ex:orthoqls}
We present a pair of $9 \times 9$ weak orthogonal quantum Latin squares, the first is the  QLS from Example~\ref{ex:qls}. Again let $\ket i$, $i\in \{0,...,9\}$ be the computational basis states and define the states $\ket a ,\ket b ,\ket c ,\ket \alpha,\ket \beta$ and $\ket \gamma$   as  in Equations~\eqref{eq:a}~\eqref{eq:b}~\eqref{eq:c}~\eqref{eq:abg}~\eqref{eq:abg2} and~\eqref{eq:abg3}. We define the following pair of QLSs:
\begin{align}
\mathcal P &:= \small\arraycolsep=3pt \grid{
 \ket{0}  &  \ket{2}  &  \ket{1}  &  \ket{3}  &  \ket{5}  &  \ket{4}  &  \ket {6}  &  \ket {8}  &  \ket 7 
\\ \hline
 \ket 2 
&  \ket{1} 
&  \ket{0} 
&  \ket 5  &  \ket 4  &  \ket 3  &  \ket 8  &  \ket 7  &  \ket 6 
\\ \hline
 \ket{1} 
&  \ket{0} 
&  \ket{2} 
&  \ket{4}  &  \ket{3}  &  \ket 5  &  \ket 7  &  \ket 6  &  \ket 8  
\\ \hline
 \ket{6}  &  \ket{8}  &  \ket{7}  &  \ket{0}  &  \ket 2  &  \ket 1  &  \ket 3  &  \ket 5  &  \ket 4 
\\ \hline
\ket 8 & \ket 7 & \ket 6 & \ket 2 & \ket 1 & \ket 0 & \ket 5 & \ket 4 & \ket 3 \\ \hline
\ket 7 & \ket 6 & \ket 8 & \ket 1 & \ket 0 & \ket 2 & \ket 4 & \ket 3 & \ket 5  \\ \hline
 \ket a & \ket c & \ket b & \ket 6 & \ket 8 & \ket 7 &  \ket \alpha  & \ket \gamma & \ket \beta  \\ \hline
\ket c & \ket b&\ket a &\ket 8&\ket 7&\ket 6&\ket \gamma&\ket \beta&\ket \alpha \\ \hline
\ket b & \ket a&\ket c &\ket 7&\ket 6&\ket 8&\ket \beta&\ket \alpha&\ket \gamma}
&
\mathcal Q &:= \small\arraycolsep=3pt\grid{
 \ket{0}  &  \ket{1}  &  \ket{2}  &  \ket{6}  &  \ket{7}  &  \ket{8}  &  \ket {3}  &  \ket {4}  &  \ket 5 
\\ \hline
 \ket 2 
&  \ket{0} 
&  \ket{1} 
&  \ket 8  &  \ket 6  &  \ket 7  &  \ket 5  &  \ket 3  &  \ket 4 
\\ \hline
 \ket{1} 
&  \ket{2} 
&  \ket{0} 
&  \ket{7}  &  \ket{8}  &  \ket 6  &  \ket 4  &  \ket 5  &  \ket 3  
\\ \hline
 \ket{a}  &  \ket{b}  &  \ket{c}  &  \ket{0}  &  \ket 1  &  \ket 2  &  \ket 6  &  \ket 7  &  \ket 8 
\\ \hline
\ket c & \ket a & \ket b & \ket 2 & \ket 0 & \ket 1 & \ket 8 & \ket 6 & \ket 7 \\ \hline
\ket b & \ket c & \ket a & \ket 1 & \ket 2 & \ket 0 & \ket 7 & \ket 8 & \ket 6  \\ \hline
 \ket 6 & \ket 7 & \ket 8 & \ket 3 & \ket 4 & \ket 5 &  \ket \alpha  & \ket \beta & \ket \gamma  \\ \hline
\ket 8 & \ket 6&\ket 7 &\ket 5&\ket 3&\ket 4&\ket \gamma&\ket \alpha&\ket \beta \\ \hline
\ket 7 & \ket 8&\ket 6 &\ket 4&\ket 5&\ket 3&\ket \beta&\ket \gamma&\ket \alpha}
\end{align}
It can be checked that if we take any row from $\mathcal P$ and any row from $\mathcal Q$ and compute the componentwise inner product of their vector entries, the resulting $n$ numbers will always be $n-1$ zeros and a single $1$.
\end{example}
\begin{theorem}\label{thm:main}
Given two indexed families of $n$ Hadamards $H_k$ and $G_j$both of size $n \times n$, and  a pair of $n \times n$ weak orthogonal quantum Latin squares $\mathcal P$ and $\mathcal Q$, the bases $B(\mathcal Q, H_k)$ and $B(\mathcal P, G_j)$ are mutually unbiased.
\end{theorem} 
\begin{proof}
See Appendix~\ref{apxprf}.
\end{proof}
\begin{example}\label{ex:qlsmub}
Given as input $\mathcal P$ and $\mathcal Q$ from Example~\ref{ex:orthoqls} and the Hadamard ${H=H_0=H_1=...=H_{n-1}}=G_0=...=G_{n-1}$ defined below with $\omega:=e^{2\pi i/3}$ we have constructed a pair of maximally entangled mutually unbiased bases $\mathcal A$ and $\mathcal B$ for the Hilbert space $\C^9 \otimes \C^9$.
\begin{equation}
H:=
\begin{pmatrix}
    1 & 1 & 1 & 1 & 1 &1&1&1&1\\
    1 & \omega & \omega^2 & 1 & \omega & \omega^2&1&\omega&\omega^2\\
1&\omega^2 &\omega&1&\omega^2 &\omega&1&\omega^2 &\omega\\
1&1&1&\omega&\omega&\omega&\omega^2&\omega^2&\omega^2\\
1&1&1&\omega&\omega&\omega&\omega^2&\omega^2&\omega^2\\
1&\omega^2 &\omega&\omega&1 &\omega^2&\omega^2 &\omega&1\\
1&1&1&\omega^2 &\omega^2 &\omega^2 &\omega&\omega&\omega\\
1&\omega&\omega^2&\omega^2&1&\omega&\omega&\omega^2&1\\
1&\omega^2 &\omega&\omega^2 &\omega&1&\omega&1 &\omega^2
 \end{pmatrix}
\end{equation}
 A sample of the $162$ basis states of $\mathcal A$ and $\mathcal B$ with some calculations showing mutual unbiasedness (see Definition~\ref{def:mub}) can be found in Appendix~\ref{apx:6561}. We have performed inner product calculations for all $6561$ combinations of states from $\mathcal A$ and $\mathcal B$ and can confirm that they are  mutually unbiased.   
\end{example}
\section{Weak orthogonality and Latin square conjugates}\label{sec:left}
In this section we explain how  weak orthogonality for QLSs restricts to orthogonality for Latin squares, and why this is the natural generalisation of orthogonality for QLSs. We start with the traditional definition of orthogonality.\begin{definition}[Orthogonal Latin squares]\label{def:ols}Given a pair of Latin squares $A$ and $B$ of equal size, we take each computational basis state from $A$  and form the ordered pair with the  state from $B$ corresponding to the same position in the grid. $A$ and $B$ are \textit{orthogonal} when this procedure gives us all possible pairs of computational basis states~\cite{mann}.   
 \end{definition} 
\noindent This definition does not lend itself to generalisation to QLSs since we may now have more than $n^2$ possible ordered pairs, but we can take an alternative approach. We characterise orthogonality in the following way:
\begin{lemma}
Latin squares $A$ and $B$ are orthogonal if and only if the following linear map $P$ is a permutation of basis states:
\begin{equation}\label{eq:orthols}
P:=\sum_{i=0}^{n-1}\sum_{j=0}^{n-1}\sum_{k=0}^{n-1}\ket i \ket j \bra{A_{ij}} \braket{k}{B_{ij}}\bra k
\end{equation}
\end{lemma} 
\begin{proof}
We now rearrange the equation defining the linear map $P$:
\begin{align*}P&:=\sum_{i}\sum_{j}\sum_{k}\ket i \ket j \bra{A_{ij}} \braket{k}{B_{ij}}\bra k\\&=\sum_{i}\sum_{j}\sum_{k}\ket i \ket j \bra{A_{ij}} \braket{B_{ij}}{k}\bra k \\ &=\sum_{i=0}^{n-1}\sum_{j=0}^{n-1}\ket i \ket j \bra{A_{ij}} \bra{B_{ij}}\sum_k \ket{k}\bra k\\ &=\sum_{i=0}^{n-1}\sum_{j=0}^{n-1}\ket i \ket j \bra{A_{ij}} \bra{B_{ij}}
\end{align*}
The second equality above holds because all $\ket{B_{ij}}$ and $\ket k$ are real valued vectors, and so $\braket{k}{B_{ij}}=\overline{\braket{k}{B_{ij}}}=\braket{B_{ij}}{k}$. The third equality is just a rearranging of terms. The last equality holds by virtue of $\sum_k \ketbra{k}{k}$ being the resolution of the identity. The linear map $P$ takes in the state $\ket p \ket q$ and outputs a superposition of all the states $\ket i  \ket j$ such that $\ket{A_{ij}}=\ket p$ \textit{and} $\ket{B_{ij}}=\ket q$, or outputs $0$ if no such $i,j$ exist. $P$ is a permutation if and only if for all inputs $p,q$ there exists unique $i,j$ such that $\ket{A_{ij}}=\ket p$ \textit{and} $\ket{B_{ij}}=\ket q$, i.e. $A$ and $B$ are orthogonal Latin squares.
\end{proof}
We now have a condition that we can apply to quantum Latin squares. However, for QLSs $A$ and $B$ this turns out to preclude superpositions, thus making $A$ and $B$ Latin squares.
\begin{lemma}
Given a pair of quantum Latin squares, if they obey equation~\eqref{eq:orthols}, then they are Latin squares.
\end{lemma}
\begin{proof}Let $A$ and $B$ be QLSs such that the linear map $P$ as defined above is a permutation of basis states. Then the adjoint of $P$, $P^\dag= \sum_{i}\sum_{j}\sum_{k}\ket{ A_{ij}} \ket k \bra{i} \braket{{B_{ij}}}{k}\bra j$ must also be a permutation of basis states. We input computational basis states $p$ and $q$ into $P^\dag$ 
\begin{align*}
P^\dag(\ket p \ket q)&=\sum_k\ket{A_{pq}}\ket k \braket{{B_{pq}}}{k}\\&=\sum_k\ket{A_{pq}}\ket k\overline{ \braket{k}{{B_{pq}}}}\\&=\sum_k\ket{A_{pq}}\ket k{ \braket{k}{\overline{{B_{pq}}}}}\\&=\ket{A_{pq}}\left [\sum_k\ket k \bra{k} \right ]\overline{\ket{{B_{pq}}}}\\&=\ket{A_{pq}}\overline{\ket{B_{pq}}}
\end{align*}
The second equality is due to the fact that the inner product is Hermitian, the third equality is due to $\ket k$ being real valued for all $k$, the fourth equality is an algebraic rearrangement and the final equality is a resolution of the identity. If $P^\dag$ above is a permutation of basis states, then for all $p,q \in \{0,...,n-1\}$, $\ket{A_{pq}}$ and $\overline{\ket{B_{pq}}}$ must be computational basis states. Thus $A$ and $B$ are Latin squares.
\end{proof}
In order to define orthogonality for QLSs we will now make a (very) brief detour into quasigroup theory. 
Latin squares can be thought of as the multiplication (Cayley) table for finite order quasigroups~\cite{smith} on the computational basis states. Let $*$ be the binary operation given by a Latin square. The fact that each state appears exactly once in each row and each column means that knowledge of any two of $a, b$ and $c$ in the equation ${a*b=c}$ uniquely determines the third. This means we can canonically define the binary operation $\backslash$ , read as \textit{left divide}, such that $a*b=c \Rightarrow a \backslash c=b$. This new binary operation defines a new quasigroup and therefore a new  Latin square called the \textit{left conjugate Latin square} (it can easily be checked that this does indeed give a Latin square)~\cite{smith}. The map that takes a Latin square and gives the left conjugate $L \super{$\backslash$}\longrightarrow L'$, is in fact involutive so we can recover $L$ from $L'$ by applying the map again. We will see a nice graphical characterisation of this fact below. The map $L \super{$\backslash$}\longrightarrow L'$ is a bijection on the set of all Latin squares. 
\begin{definition}[Left orthogonality]\label{def:lols}
Given a pair of Latin squares they are \textit{left orthogonal} when their left conjugates are orthogonal.
\end{definition}
\begin{remark}We could equally well talk about the right conjugate given by right divide and define right orthogonality. In this paper we only make use of left orthogonality.
\end{remark}

Since $L \super{$\backslash$}\longrightarrow L'$ is a bijection as mentioned above, the set of orthogonal Latin squares and left orthogonal Latin squares are isomorphic. Left orthogonality is in fact the property that we have generalised to QLSs in Definition~\ref{def:woqls}. 

To proceed further it will be useful to introduce some diagrams (see Appendix~\ref{apx:cqm}). Let \tinymult[lsdot] be a Latin square and $\tinyspider$ be the classical structure corresponding to the computational basis. Then the left divide map has the following form:
\begin{equation}\label{eq:leftdivide2}
\begin{pic}
   \node[lsdot] (ls) {};
   \draw[string,in=right,out=up] (0.25,-0.5) to (ls.center);
   \draw[string,in=left,out=up] (-0.25,-0.5) to (ls.center);
   \draw[string] (ls.center) to (0,0.5);
\end{pic}
\quad \super{$\backslash$}\longrightarrow \quad
\begin{pic}
   \node[lsdot] (ls) {};
   \node[blackdot,scale=0.8] (b)  at (-0.35,0.25) {};
   \draw[string] (ls.center) to (0,-0.5);
   \draw[string,in=right,out=down] (0.25,0.5) to (ls.center);
   \draw[string,out=left,in=right,looseness=1.7] (ls.center) to (b.center);
   \draw[string,in=left,out=up] (-0.7,-0.5) to (b.center);
\end{pic}
\end{equation}\label{eq:leftdivide} 
The fact that $\backslash$ is an involution can be verified using the snake equation:
\begin{equation}\label{eq:idem}
\begin{pic}
   \node[lsdot] (ls) {};
   \draw[string,in=right,out=up] (0.25,-0.5) to (ls.center);
   \draw[string,in=left,out=up] (-0.25,-0.5) to (ls.center);
   \draw[string] (ls.center) to (0,0.5);
\end{pic}
\quad \super{$\backslash$}\longrightarrow \quad
\begin{pic}
   \node[lsdot] (ls) {};
   \node[blackdot,scale=0.8] (b)  at (-0.35,0.25) {};
   \draw[string] (ls.center) to (0,-0.5);
   \draw[string,in=right,out=down] (0.25,0.5) to (ls.center);
   \draw[string,out=left,in=right,looseness=1.7] (ls.center) to (b.center);
   \draw[string,in=left,out=up] (-0.7,-0.5) to (b.center);
\end{pic}
\quad \super{$\backslash$}\longrightarrow \quad
\begin{pic}
   \node[lsdot] (ls) {};
   \node[blackdot,scale=0.8] (b)  at (-0.35,-0.25) {};
   \node[blackdot,scale=0.8] (b2) at (-0.7,0){};
   \draw[string] (ls.center) to (0,0.5);
   \draw[string,in=right,out=up] (0.25,-0.5) to (ls.center);
   \draw[string,out=left,in=right,looseness=1.7] (ls.center) to (b.center);
   \draw[string,in=left,out=right,looseness=1.7] (b2.center) to (b.center);
   \draw[string,in=left,out=up] (-0.95,-0.5) to (b2.center);
\end{pic}
\quad \super{\eqref{eq:sm}}= \quad
\begin{pic}
   \node[lsdot] (ls) {};
   \draw[string,in=right,out=up] (0.25,-0.5) to (ls.center);
   \draw[string,in=left,out=up] (-0.25,-0.5) to (ls.center);
   \draw[string] (ls.center) to (0,0.5);
\end{pic}
\end{equation} 
For Latin squares $A=\tinymult[lsdot]$ and $B=\tinymult[lssdot]$, equation~\eqref{eq:orthols} can be expressed diagramatically as follows:
\begin{equation}\label{eq:orthoperm}
\begin{pic}
   \node[morphism,scale=2](p) {$P$};
   \draw[string] (0.475,0.6) to (0.475,1.2);
   \draw[string] (-0.475,0.6) to (-0.475,1.2);
   \draw[string] (0.475,-0.6) to (0.475,-1.2);
   \draw[string] (-0.475,-0.6) to (-0.475,-1.2);\end{pic}
\quad := \quad
\begin{pic}
   \node[lsdot] (ls) at (0,-0.25){};  
   \node[lssdot] (lss) at (0,1){};
   \node[blackdot,scale=0.8] (b1) at(-0.5,0.375){};
   \node[blackdot,scale=0.8] (b2) at (0.5,0.375){};
   \node[blackdot,scale=0.8] (b3) at (0.4,1.3){};
   \draw[string, in=down,out=left] (ls.center) to (b1.center);
   \draw[string, in=down,out=right] (ls.center) to (b2.center);    
   \draw[string, in=right,out=up] (b2.center) to (lss.center);
   \draw[string, in=up,out=left] (lss.center) to (b1.center);
   \draw[string, in=left,out=up] (lss.center) to (b3.center); 
   
   \node (1) at (0,-0.8){};
   \draw[string] (1.center) to (ls.center);  
   \node (2) at (-1,1.8){};
   \draw[string] (2) to (b1.center); 
   \node (3) at (1,1.8){};
   \draw[string] (3) to (b2.center);         
   \node (4) at (1.25,-0.8){};
   \draw[string,in=right,out=up] (4.center) to (b3.center);
\end{pic}
\quad
\text{ is a permutation}
\end{equation}
We now  substitute in the left conjugates of Latin squares $A$ and $B$,  $\tinymult[lsdot]\super{$\backslash$}\longrightarrow \tinyleftdivide[lsdot]$ and  $\tinymult[lssdot]\super{$\backslash$}\longrightarrow \tinyleftdivide[lssdot]$  to obtain a linear map $P'$ which must be a permutation of basis states for $A$ and $B$ to be left orthogonal.
The condition that $A$ and $B$ are left orthogonal is thus equivalent to the following statement:
\begin{equation}\label{eq:lolsb}
\begin{pic}
   \node[morphism,scale=2](p) {$P'$};
   \draw[string] (0.475,0.6) to (0.475,1.2);
   \draw[string] (-0.475,0.6) to (-0.475,1.2);
   \draw[string] (0.475,-0.6) to (0.475,-1.2);
   \draw[string] (-0.475,-0.6) to (-0.475,-1.2);
\end{pic}
\quad := \quad
\begin{pic}
   \node[lsdot] (ls) at (0,-0.25){};  
   \node[lssdot] (lss) at (0,1){};
   \node[blackdot,scale=0.8] (b1) at(-0.75,0.375){};
   \node[blackdot,scale=0.8] (b2) at (0,0.375){};
   \node[blackdot,scale=0.8] (b3) at (0.4,1.3){};
   \node[blackdot,scale=0.8] (b5) at (-0.375,1.25){};
   \node[blackdot,scale=0.8] (b4) at (-0.375,-0.5){};
   \draw[string,out=left,in=up] (b5.center) to (b1.center);
   \draw[string, in=right,out=left] (ls.center) to (b4.center);
   \draw[string,out=left,in=down] (b4.center) to (b1.center);
   \draw[string, in=down,out=up] (ls.center) to (b2.center);    
   \draw[string, in=down,out=up] (b2.center) to (lss.center);
   \draw[string, in=right,out=left] (lss.center) to (b5.center);
   \draw[string, in=left,out=right,looseness=1.5] (lss.center) to (b3.center); 
   
   \node (1) at (0.25,-0.8){};
   \draw[string,in=right,out=up] (1.center) to (ls.center);  
   \node (2) at (-1,1.8){};
   \draw[string] (2.center) to (b1.center); 
   \node (3) at (1,1.8){};
   \draw[string] (3.center) to (b2.center);         
   \node (4) at (1.25,-0.8){};
   \draw[string,in=right,out=up] (4.center) to (b3.center);
\end{pic}
\quad \super{\eqref{eq:sm}}= \quad
\begin{pic}
   \node[lsdot] (ls) at (0,-0.25){};  
   \node[lssdot] (lss) at (0,1){};
   \node (b1) at(-0.75,0.375){};
   \node[blackdot,scale=0.8] (b2) at (0,0.375){};
   \node[blackdot,scale=0.8] (b3) at (0.4,1.3){};
   \node[blackdot,scale=0.8] (b5) at (-0.375,1.25){};
   \node[blackdot,scale=0.8] (b4) at (-0.375,-0.5){};
   \draw[string,out=left,in=up] (b5.center) to (b1.center);
   \draw[string, in=right,out=left] (ls.center) to (b4.center);
   \draw[string,out=left,in=down] (b4.center) to (b1.center);
   \draw[string, in=down,out=up] (ls.center) to (b2.center);    
   \draw[string, in=down,out=up] (b2.center) to (lss.center);
   \draw[string, in=right,out=left] (lss.center) to (b5.center);
   \draw[string, in=left,out=right,looseness=1.5] (lss.center) to (b3.center); 
   
   \node (1) at (0.25,-0.8){};
   \draw[string,in=right,out=up] (1.center) to (ls.center);  
   \node (2) at (-0.375,1.8){};
   \draw[string] (2.center) to (b5.center); 
   \node (3) at (1,1.8){};
   \draw[string] (3.center) to (b2.center);         
   \node (4) at (1.25,-0.8){};
   \draw[string,in=right,out=up] (4.center) to (b3.center);
\end{pic}
\text{ is a permutation}
\end{equation} 
In words: first we input two states $i$ and $j$ and then compute the component-wise  inner products of the $i^{th}$ row of $A$ and the $j^{th}$ row of $B$. There must be one unique  column, say $s,$ such that $\braket{B_{sj}}{A_{si}}=1$ with $\braket{B_{rj}}{A_{ri}}=0$ for all $r$ not equal to $s$. We then output $s$ on  the left and $\ket{A_{si}}$  on the right. The set of output states $s \otimes \ket{A_{si}}$ must be every possible combination of computational basis states. 

We can interpret this for QLSs but again we encounter the same difficulty. \begin{lemma}
Every pair of left orthogonal QLSs are Latin squares.
\end{lemma}
\begin{proof}
For a contradiction assume that $A$ and $B$ are left orthogonal QLSs that are not Latin squares. There is some vector  entry in $A$ that is not a computational basis state say $\ket{A_{pq}}$. For $P'$ as defined in Equation~\eqref{eq:lolsb} to be a permutation, $\ket{A_{pq}}$ cannot be the output on the right for any input $q,j$. This means that no row of $B$  has the complex conjugate of  $\ket{A_{pq}}$ as its $p^{th}$ column entry. But each row of $B$ must have one column entry that is the complex conjugate of the corresponding column entry of the $q^{th}$ row of $A.$ Thus at least two of the rows of $B$\ have the same vector in the same column. This violates the rule that $B$ is a QLS and thus gives a contradiction. Therefore $A$ must be a Latin square. Reversing the roles, we find that $B$ must be a Latin square too (left orthogonality, like orthogonality is a symmetric relation).
 \end{proof}
 The condition must therefore be weakened if we want to define a property that non-Latin square QLSs can satisfy. One approach is to delete the output from the right hand wire and require that the linear map thus obtained be a function on the computational basis states. This is in fact the \textit{weak orthogonality} property of Definition~\ref{def:woqls}. This condition turns out to be strong enough to give rise to interesting and useful properties such as using QLSs to build mutually unbiased MEBs (see Theorem \ref{thm:main}), yet weak enough so that pairs of Latin squares are weak orthogonal if and only if they are orthogonal.

Diagrammatically Definition~\ref{def:woqls} becomes the following:
\begin{equation}\label{eq:lols}
\begin{pic}
   \node[whitedot,scale=2](p) {$f$};
   \draw[string] (p.center) to (0,1.2);
   \draw[string] (0.475,0) to (0.475,-1.2);
   \draw[string] (-0.475,0) to (-0.475,-1.2);
\end{pic}
\quad := \quad
\begin{pic}
   \node[lsdot] (ls) at (0,-0.25){};  
   \node[lssdot] (lss) at (0,1){};
   \node (b1) at(-0.75,0.375){};
   \node[blackdot,scale=0.8] (b2) at (0,0.375){};
   \node[blackdot,scale=0.8] (b3) at (0.4,1.3){};
   \node[blackdot,scale=0.8] (b5) at (-0.375,1.25){};
   \node[blackdot,scale=0.8] (b4) at (-0.375,-0.5){};
   \draw[string,out=left,in=up] (b5.center) to (b1.center);
   \draw[string, in=right,out=left] (ls.center) to (b4.center);
   \draw[string,out=left,in=down] (b4.center) to (b1.center);
   \draw[string, in=down,out=up] (ls.center) to (b2.center);    
   \draw[string, in=down,out=up] (b2.center) to (lss.center);
   \draw[string, in=right,out=left] (lss.center) to (b5.center);
   \draw[string, in=left,out=right,looseness=1.5] (lss.center) to (b3.center); 
   
   \node (1) at (0.25,-0.8){};
   \draw[string,in=right,out=up] (1.center) to (ls.center);  
   \node (2) at (-0.375,1.8){};
   \draw[string] (2.center) to (b5.center); 
   \node[blackdot,scale=0.8] (3) at (0.8,1.5){};
   \draw[string] (3.center) to (b2.center);         
   \node (4) at (1.25,-0.8){};
   \draw[string,in=right,out=up] (4.center) to (b3.center);
\end{pic}
\quad \super{\eqref{eq:sm}}= \quad
\begin{pic}
   \node[lsdot] (ls) at (0,-0.25){};  
   \node[lssdot] (lss) at (0,1){};
   \node (b1) at(-0.75,0.375){};
   \node[blackdot,scale=0.8] (b3) at (0.4,1.3){};
   \node[blackdot,scale=0.8] (b5) at (-0.375,1.25){};
   \node[blackdot,scale=0.8] (b4) at (-0.375,-0.5){};
   \draw[string,out=left,in=up] (b5.center) to (b1.center);
   \draw[string, in=right,out=left] (ls.center) to (b4.center);
   \draw[string,out=left,in=down] (b4.center) to (b1.center);
   \draw[string, in=down,out=up] (ls.center) to (b2.center);    
   \draw[string, in=down,out=up] (b2.center) to (lss.center);
   \draw[string, in=right,out=left] (lss.center) to (b5.center);
   \draw[string, in=left,out=right,looseness=1.5] (lss.center) to (b3.center); 
   
   \node (1) at (0.25,-0.8){};
   \draw[string,in=right,out=up] (1.center) to (ls.center);  
   \node (2) at (-0.375,1.8){};
   \draw[string] (2.center) to (b5.center); 
   \node (3) at (1,1.8){};
   \node (4) at (1.25,-0.8){};
   \draw[string,in=right,out=up] (4.center) to (b3.center);
\end{pic}
\text{ is a function}
\end{equation}

\begin{lemma}\label{lem:rest}
Given a pair of Latin squares, $A$ and $B$ the following are equivalent:
\begin{itemize}
   \item $A$ and $B$ are weak orthogonal (see Definition~\ref{def:woqls}).   \item $A$ and $B$ are left orthogonal (see Definition~\ref{def:lols}).
\end{itemize}
\end{lemma} 
\begin{proof}
If $A$ and $B$ are left orthogonal then $P'$, as defined in Equation~\eqref{eq:lolsb}, is a permutation of basis states, which clearly implies the weaker condition that $f$ as defined in Equation~\eqref{eq:lols} is a function. For the other implication let $A$ and $B$ be weak orthogonal Latin squares. Consider the $p^{th}$ columns of $A$ and $B.$ They both contain all $n$ computational basis states and there must therefore exist values of $i$ and $j$ for all $q\in \{0,...,n-1\}$ such that $\ket{A_{pi}}=\ket{B_{pj}}=\ket q$. So for column $p$ there exist $i,j$ such that  $P'(\ket i \otimes \ket j)=\ket p \otimes\ket q$ for all $q$. This is true for all rows $q$, so $P'$ is a permutation. 
\end{proof} 
\begin{remark}
We defined weak orthogonality from left orthogonality by setting the requirement  that the linear map $P'$ (see Equation~\eqref{eq:lolsb})with the right hand output deleted needs to be a function on the basis states, rather than requiring $P'$ itself to be a permutation of the basis states. We could have tried to weaken orthogonality directly by requiring that $P$ (see Equation~\eqref{eq:orthoperm}) with the right hand output deleted be a function on basis states. However, it turns out that this would still  preclude non-Latin square QLSs.   
\end{remark}
\section{Beth and Wocjan's MUB construction}\label{sec:bw}

In their 2004 paper~\cite{ortho1} Beth and Wocjan gave a construction for a pair of mutually unbiased bases of a  Hilbert space $\mathcal H$ of square dimension $s=n^2$,  given as input a pair of  $n \times n$ orthogonal Latin squares and an $n \times n$ Hadamard matrix which was later put in explicit Latin square form by Wehner and Winter~\cite{ortho1,ortho2}.

The construction takes each Latin square together with the Hadamard and produces an MEB of dimension $n^2$. The fact that the Latin squares are orthogonal is then shown to entail that these two bases are mutually unbiased. I will refer to this MEB construction as the Left Beth-Wocjan maximally entangled basis (LBW MEB) construction\footnote{The construction presented here is technically the construction given by taking the left conjugate of the Latin square $L$ first and then applying the construction defined by Beth and Wocjan. Since taking the left conjugate gives us a bijection (see Equation~\eqref{eq:leftdivide}) on the set of Latin squares the MEBs obtainable are not affected by this.}.

\begin{definition}[Left Beth-Wocjan maximally entangled basis]\label{def:bw}
Given an $n \times n$ Latin square $L$ and an $n\times n$ Hadamard $H$, then $\mathcal B$ as  defined below is a \textit{Left Beth-Wocjan maximally entangled basis} (LBW MEB).~\footnote{~\label{ftn}The definition below is slightly different to the one given by Beth and Wocjan even taking into account the use of the left conjugate Latin square. However, when the input is a Latin square the two constructions agree precisely. }\begin{equation}\label{eq:bwmeb}
\mathcal B := \left\{ B_{ij} = \frac{1}{\sqrt{n}}\sum_{k,p=0}^{n-1}\ket {k,p}H_{ik}\braket{L_{kp}} j \text{ such that } i,j \in \{ 0,..,n-1\} \right\}
\end{equation} 
\end{definition} 

The graphical calculus gives a good notation with which to compare LBW MEBs to  QLS MEBs (see Definition~\ref{def:qlsmeb}). 
\begin{lemma}\label{lem:subset}
Under the restriction to Latin squares and to having a single fixed Hadamard the  QLS MEBs  are the same as LBW MEBs. 
\end{lemma}
\begin{proof}We construct an LBW MEB $B_{ij}$ and a QLS MEB $C_{ij}$ from the latin square $L=\tinymult[lsdot]$ and Hadamard $H$.

\vspace{20pt}
\hspace{0pt}\begin{tabular}[b]{m{6cm} m{7cm}}
\textit{\textbf{Left Beth-Wocjan MEB}} &
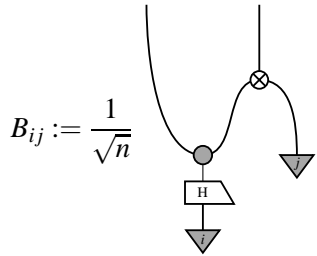
\textit{\textbf{Quantum Latin square MEB}} \\
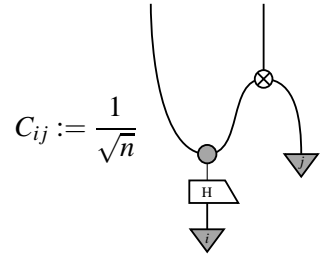
\begin{equation*}\label{eq:bwdiag}\vspace{15pt}B_{ij}:=
\frac{1}{\sqrt{n}} \begin{pic}
          \node (in) at (-0.5,4) {};         
          \node (H)[morphism,wedge,scale=0.5] at(0.25,1.5) {H};
          \node (b1)[blackdot,scale=1] at (0.25,2) {};
          \node (b2)[lsdot,scale=1.2] at (1,3) {};          
          \node (i)[state,black,scale=0.5] at(0.25,1) {$i$};
          \node (0b) at (0.5,2) {};
          \node (j)[state,black,scale=0.5] at(1.5,2) {$j$};
          \node (out) at (1,4) {}; 
          \draw[string,out=270,in=180,looseness=0.75] (in.center) to (b1.center);           \draw (H.north) to (b1.center);
          \draw[string,out=0,in=180] (b1.center) to (b2.center);
          \draw[string,out=90,in=0] (j.center) to (b2.center);
          \draw[string] (b2.center) to (out.center);
          \draw[string] (i) to (H.south);
\end{pic}\end{equation*} &
\begin{equation*}\label{eq:qlsdiag}C_{ij}:=
\frac{1}{\sqrt{n}}\begin{pic}
          \node (in) at (-0.5,4) {};         
          \node (H)[morphism,wedge,scale=0.5] at(0.25,1.5) {H};
          \node (b1)[blackdot,scale=1] at (0.25,2) {};
          \node (b2)[lsdot,scale=1.2] at (1,3) {};          
          \node (i)[state,black,scale=0.5] at(0.25,1) {$i$};
          \node (0b) at (0.5,2) {};
          \node (j)[state,black,scale=0.5] at(1.5,2) {$j$};
          \node (out) at (1,4) {}; 
          \draw[string,out=270,in=180,looseness=0.75] (in.center) to (b1.center);           \draw (H.north) to (b1.center);
          \draw[string,out=0,in=180] (b1.center) to (b2.center);
          \draw[string,out=90,in=0] (j.center) to (b2.center);
          \draw[string] (b2.center) to (out.center);
          \draw[string] (i) to (H.south);
          \end{pic}\end{equation*}
\\
\end{tabular}

We see that the diagrams are the same.\end{proof}
\begin{theorem}\label{thm:bw}
Given a pair of $n \times n$ left\footnote{In their paper Beth and Wocjan use orthogonal Latin squares, but since we defined their MEB construction on the left conjugate the \textit{left} becomes necessary here.} orthogonal Latin squares and an $n \times n$ Hadamard,  construct two LBW MEBs using each Latin square with the Hadamard. The bases are mutually unbiased.  
\end{theorem}
\begin{lemma}
The construction of MUBs in Theorem~\ref{thm:main} restricts to the construction of Theorem~\ref{thm:bw}, under the restriction of the QLS to a Latin square and the two families of Hadamards to a single fixed Hadamard. 
\end{lemma}
\begin{proof}
Follows directly from Lemma~\ref{lem:subset}. 
\end{proof}
The following  corollary gives a construction for MUBs in square dimension that is more general than the LBW MUB construction but not as general as our main construction. \begin{corollary}\label{corl}
Given two indexed families of $n$ Hadamards $H_k$ and $G_j$ both of size $n \times n$, and  a pair of $n \times n$ left orthogonal  Latin squares $\mathcal P$ and $\mathcal Q$, the bases $B(\mathcal P, H_k)$ and $B(\mathcal Q, G_j)$ are mutually unbiased.
\end{corollary}

 So our new construction generalises Beth and Wocjan's in two directions, having two arbitrary families of Hadamards rather than a single fixed Hadamard and quantum Latin squares
rather than Latin squares. The next theorem shows, by explicit example, that the generalisation is strict.

\begin{theorem}\label{thm:ineq}
The pair of mutually unbiased MEBs from Example~\ref{ex:qlsmub}  are inequivalent to any MEBs obtainable by the LBW MEB construction.
\end{theorem} 
\begin{proof}
It will be sufficient to prove that one of our MEBs is inequivalent to any obtainable by the LBW MEB construction. Since equivalence of MEBs is the same as equivalence of UEBs we will take the dual approach here (see Section~\ref{sec:dual} below) and prove that the UEB arising from QLS $\mathcal P$ and Hadamard $H$ in Example~\ref{ex:qlsmub}, which we will refer to as $X$,  is inequivalent to any  LBW UEB. 

We will proceed along the same lines as~\cite[Corollary 31]{mypaper1}. Note that LBW UEBs are a restriction to a single fixed Hadamard of shift-and-multiply UEBs. Thus by~\cite[Proposition 30]{mypaper1}, LBW UEBs are \textit{monomial} (meaning each unitary matrix of the basis is the product of a diagonal matrix and a permutation matrix). 

Suppose for a contradiction  that $X$ is equivalent to a monomial basis.  The first matrix of $X$ is as follows:
\begin{equation*}
X_{00}=
\begin{pmatrix}
1&0&0&0&0&0&0&0&0\\
0&0&1&0&0&0&0&0&0\\
0&1&0&0&0&0&0&0&0\\
0&0&0&1&0&0&0&0&0\\
0&0&0&0&0&1&0&0&0\\
0&0&0&0&1&0&0&0&0\\
0&0&0&0&0&0&1&0&0\\
0&0&0&0&0&0&0&0&1\\
0&0&0&0&0&0&0&1&0
\end{pmatrix}
\end{equation*}
$X_{00}$ is self adjoint. We obtain the equivalent  UEB $X'$ by composing all the matrices of $X$ on the right by $X_{00}$. Thus $X'_{00}=\id[9]$.
Now $X'$ contains the identity and is equivalent to a monomial basis so by~\cite[Proposition 26]{mypaper1} $X'$ is \textit{simultaneously monomializable.} (See~\cite[Definition 25]{mypaper1} .  
The least common multiple of $\{1,2,3,4,5,6,7,8,9\}$ is $\mu_9 = 2520$; thus by~\cite[Proposition 28]{mypaper1}  the $2520^{th}$ powers of the elements of $X$ will commute. Now let $\omega=e^{2\pi i/3}$and  consider $X'_{06}$ (left) and $X'_{07}$ (right) below:

\begin{equation*}
\begin{pmatrix}
0&0&0&0&0&0& \frac{1}{\sqrt{3}}&\frac{1}{\sqrt{3}}&\frac{1}{\sqrt{3}}\\
0&0&0&0&0&0& \frac{1}{\sqrt{3}}&\frac{\omega^2}{\sqrt{3}}&\frac{\omega}{\sqrt{3}}\\
0&0&0&0&0&0& \frac{1}{\sqrt{3}}&\frac{\omega}{\sqrt{3}}&\frac{\omega^2}{\sqrt{3}}\\
\frac{1}{\sqrt{3}}&-i\sqrt{\frac{2}{7}}&\sqrt{\frac{2}{3}}&0&0&0&0&0&0\\
\frac{1}{\sqrt{3}}&\frac{-i}{\sqrt{14}}&\frac{-1}{\sqrt{6}}&0&0&0&0&0&0\\
\frac{i}{\sqrt{3}}&\frac{3}{\sqrt{14}}&\frac{i}{\sqrt{6}}&0&0&0&0&0&0\\
0&0&0&1&0&0&0&0&0\\
0&0&0&0&0&1&0&0&0\\
0&0&0&0&1&0&0&0&0
\end{pmatrix}
,
\begin{pmatrix}
0&0&0&0&0&0& \frac{1}{\sqrt{3}}&\frac{1}{\sqrt{3}}&\frac{1}{\sqrt{3}}\\
0&0&0&0&0&0& \frac{\omega^2}{\sqrt{3}}&\frac{\omega}{\sqrt{3}}&\frac{1}{\sqrt{3}}\\
0&0&0&0&0&0& \frac{\omega}{\sqrt{3}}&\frac{\omega^2}{\sqrt{3}}&\frac{1}{\sqrt{3}}\\
-i\sqrt{\frac{2}{7}}&\sqrt{\frac{2}{3}}&\frac{1}{\sqrt{3}}&0&0&0&0&0&0\\
\frac{-i}{\sqrt{14}}&\frac{-1}{\sqrt{6}}&\frac{1}{\sqrt{3}}&0&0&0&0&0&0\\
\frac{3}{\sqrt{14}}&\frac{i}{\sqrt{6}}&\frac{i}{\sqrt{3}}&0&0&0&0&0&0\\
0&0&0&1&0&0&0&0&0\\
0&0&0&0&0&1&0&0&0\\
0&0&0&0&1&0&0&0&0
\end{pmatrix}
\end{equation*}
For a contradiction we now compute the first column first row entry of the  commutator:
\begin{align*}K&:=(X'_{06})^{2520}(X'_{07})^{2520}-(X'_{07})^{2520}(X'_{06})^{2520} \\ \bra 0 K&\ket 0\approx-0.0219 + 0.0252i 
\neq 0
\end{align*}  
Thus $X'$ and therefore $X$ is not equivalent to any monomial basis, and in particular any LBW MEB.
\end{proof}
\section{Mutually unbiased error bases}\label{sec:dual}
Unitary error bases (UEBs) are the mathematical data necessary for protocols such as dense coding and teleportation as well as having important applications to quantum error correction. In this section we explain how the results of this paper can also be described in terms of UEBs via the correspondence between maximally entangled bases in square dimension  and UEBs by introducing the natural concept of mutually unbiased UEBs. 
\begin{definition}[Unitary error basis]
A \textit{unitary error basis} on an $n$-dimensional Hilbert space is a family of $n^{2}$ unitary matrices $U_i$, each of size $n \times n$, such that~\cite{klapp}:
\begin{equation}\label{eqdef:ueb}
 \tr (U^{\dag}_i \circ U_{j})=\delta_{ij}n 
\end{equation}
\end{definition}

Via state-process duality a bijection exists between UEBs and MEBs (See Definition~\ref{def:meb})~\cite{ghosh}. The correspondence is particularly clear diagrammatically.

Given a UEB, $\mathcal A:=\{U_{i}|0<i \leq n^2\}$ and the computational basis $\tinyspider$, we have the corresponding MEB, $\mathcal B :=\{\ket{U_{i}}|0<i \leq n^2\}$   defined as follows (see~\cite{werner2000} Lemma 2):
\begin{equation}\label{uebmeb}
U_i := \quad 
\begin{pic}
\node[morphism,scale=0.7] (U) {$U_i$};
\draw[string] (0,-0.5) to (U.south);
\draw[string] (0,0.5) to (U.north); 
\end{pic}
\quad \leadsto \quad
\frac{1}{\sqrt{n}}\begin{pic}
\node[blackdot] (b) {};
\node (1) at (-0.5,1) {};
\node (2) at (0.5,1) {};
\node[morphism,scale=0.7] (U) at (0.5,0.5) {$U_i$};
\draw[string,in=left,out=down,looseness=1.3] (1.center) to (b.center);
\draw[string,in=up,out=down,looseness=1.3] (2.center) to (U.north);
\draw[string,in=down,out=right,looseness=1.1] (b.center) to (U.south);
\end{pic}
\quad =:\ket{U_i}
\end{equation}
By Equation~\eqref{mes} the condition that the  matrices $U_i$ are unitary means that the states $\ket{U_i}$ are maximally entangled. Under this duality equivalence of MEBs as described by Equation~\ref{eq:eqimeb}, becomes the usual notion of equivalence for UEBs. The fact that the states on the right hand side of Equation~\eqref{uebmeb} are orthonormal follows directly from Equation~\eqref{eqdef:ueb} as follows:
\begin{equation}\label{eq:dual}
\braket{U_i}{U_j}\super{\eqref{uebmeb}}= 
\frac{1}{n}\begin{pic}
\node[blackdot] (b) at (0,-1) {};
\node (1) at (-0.5,0) {};
\node (2) at (0.5,0) {};
\node[morphism,scale=0.7] (U) at (0.5,-0.5) {$U_i$};
\draw[string,in=left,out=down,looseness=1.3] (1.center) to (b.center);
\draw[string,in=up,out=down,looseness=1.3] (2.center) to (U.north);
\draw[string,in=down,out=right,looseness=1.1] (b.center) to (U.south);

\node[blackdot] (b1) at (0,0.75) {};
\node (a1) at (-0.5,0) {};
\node (a2) at (0.5,0) {};
\node[morphism,scale=0.7] (aU) at (0.5,0.25) {$U^\dagger_j$};
\draw[string,in=left,out=up,looseness=1.3] (a1.center) to (b1.center);
\draw[string,in=down,out=up,looseness=1.3] (a2.center) to (aU.south);
\draw[string,in=up,out=right,looseness=1.1] (b1.center) to (aU.north);

\end{pic}
=\frac{1}{n}\tr (U^{\dag}_i \circ U_{j})\super{\eqref{eqdef:ueb}}=\delta_{ij}
\end{equation} 
In this paper the dual MEB constructions of two of the main constructions for UEBs were used. As mentioned above Lemma~\ref{lem:qlsmeb} the QLS MEB\ of that lemma is the dual of the quantum shift-and-multiply error bases of this author's paper with Jamie Vicary~\cite{mypaper1}. The MEB used in Corollary~\ref{corl} is the shift-and-multiply basis introduced by Werner~\cite{werner2001all}. Thus the LBW MEB construction described in Definition~\ref{thm:bw} gives us a family of UEBs strictly contained within Werners construction.

The duality of MEBs and UEBs makes it natural to talk about mutually unbiased unitary error bases. \begin{definition}[Mutually unbiased error bases]
A pair of unitary error bases over a Hilbert space $\mathcal H$ of dimension $n$, ${\mathcal A=\{U_{i}|i \in \{ 0,...,n-1\}\}}$ and ${\mathcal B=\{V_j|j \in \{ 0,...,n-1\}\}}$ are \textit{mutually unbiased} when the following equation holds for all $i,j$:
\begin{equation}
|\tr (U^{\dag}_i \circ V_j)| ^2= \frac{1}{n}
\end{equation}
 \end{definition}
We had two choices in defining mutually unbiased UEBs above, we  used the inner product of Equation~\eqref{eqdef:ueb} to interpret Equation~\eqref{def:mub} of Definition~\ref{def:mub} directly but we could have defined mutually unbiased UEBs to be UEBs with corresponding MEBs that are mutually unbiased. Fortunately it does not matter  as they are equivalent by a similar argument to Equation~\eqref{eq:dual}.

This definition brings up the question of what it may mean for two teleportation protocols to be mutually unbiased, or what kind of error correction could be performed by a pair of mutually unbiased error bases. 

The main result of this paper can now be interpreted as a construction for a pair of mutually unbiased unitary error bases from a pair of weak orthogonal quantum Latin squares. 
\section{Mutually orthogonal quantum Latin squares}\label{sec:moqls}
In this section we introduce the concept of families of orthogonal quantum Latin squares. In their 2004 paper  Beth and Wocjan~\cite{ortho1} introduced the construction of square dimensional MUBs  from orthogonal Latin squares as described in Section~\ref{sec:bw}. They used this construction to improve the known lower bounds for maximal sets of pairwise mutually unbiased bases. A set of \textit{mutually orthogonal Latin squares} (MOLs) is a set of two or more Latin squares that are pairwise  orthogonal. Beth and Wocjan use their construction on a set of $w$ MOLs of size $n\times n$ and give $w+2$ MUBs for dimension $n^2$. The extra two MUBs come from the  two  squares of vectors (which do not satisfy the axioms to be Latin squares, or even  quantum Latin squares) described below:~\footnote{Note that due to the presentation of Beth and Wocjan's construction in Section~\ref{sec:bw}, in which we start by taking the left-conjugate, the left conjugate map must also be applied to these  squares of vectors to recover the ones used by Beth and Wocjan. In addition the second square here only gives a basis using the original Beth-Wocjan method and not the altered version given by definition~\ref{def:bw} (See footnote~\ref{ftn}).}

\begin{itemize}
\item The first is the $n \times n$ grid with the $i^{th}$ row consisting of the repeated entry $\ket i$ for every column.

\item The second is the $n \times n$ grid with
 $\sum_k^{n-1}\ket k$ as every diagonal entry and $0$s elsewhere.  
\end{itemize}
Some thought reveals that although they are not Latin squares, these two squares are left orthogonal to every $n \times n$ Latin square and to each other. Note that the bases obtained from these extra two however are not maximally entangled. The following definition is a natural extension of the concept of sets of MOLs.
\begin{definition}[Mutually weak orthogonal quantum Latin squares]
A set of $w$ quantum Latin squares are \textit{Mutually weak orthogonal quantum Latin squares} (MOQLs) when they are pairwise weak orthogonal.
\end{definition}
There are no generalisations of the two  squares of vectors described above  that would be weak orthogonal to every QLS. However, with a particular set of MOQLs, an analogue of the first vector square above can be found by considering the subspaces spanned by the non-computational basis states. As an example we present a square of vectors that is weak orthogonal   to both of the pair of  weak orthogonal QLSs from Example~\ref{ex:orthoqls}. Again let $\ket i$, $i\in \{0,...,9\}$ be the computational basis states and define the states $\ket a , \ket b , \ket c,\ket \alpha,\ket \beta$ and $\ket \gamma$   as  in Equations~\eqref{eq:a}~\eqref{eq:b}~\eqref{eq:c}~\eqref{eq:abg}~\eqref{eq:abg2} and~\eqref{eq:abg3}. We define the following  square of vectors:
\begin{equation*}\small\arraycolsep=3pt \grid{
 \ket{0}  &  \ket{0}  & \ket{0}  & \ket{0}  & \ket{0}  & \ket{0}  & \ket{\alpha} & \ket{\alpha}  & \ket{\alpha}
\\ \hline
 \ket 1 
&  \ket{1} 
&  \ket{1} &  \ket{1} &  \ket{1} &  \ket{1} &  \ket{\beta} &  \ket{\beta} & \ket{\beta}
\\ \hline
 \ket{2} 
&  \ket{2} 
&  \ket{2} 
&  \ket{2}  &  \ket{2}  &  \ket 2  &  \ket \gamma  &  \ket \gamma  &  \ket \gamma  
\\ \hline
 \ket{a}  &  \ket{a}  &  \ket{a}  &  \ket{3}  &  \ket 3  &  \ket 3  &  \ket 3  &  \ket 3  &  \ket 3 
\\ \hline
\ket b & \ket b & \ket b & \ket 4 & \ket 4 & \ket 4 & \ket 4 & \ket 4 & \ket 4 \\ \hline
\ket c & \ket c & \ket c & \ket 5 & \ket 5 & \ket 5 & \ket 5 & \ket 5 & \ket 5  \\ \hline
 \ket 6 & \ket 6 & \ket 6 & \ket 6 & \ket 6 & \ket 6 &  \ket 6  & \ket 6 & \ket 6  \\ \hline
\ket 7 & \ket 7 &\ket 7 &\ket 7 &\ket 7 &\ket 7 &\ket 7 &\ket 7 &\ket 7 \\ \hline
\ket 8 & \ket 8 &\ket 8 &\ket 8 &\ket 8 &\ket 8 &\ket 8 &\ket 8 &\ket 8}
\end{equation*}
It can be checked that this square is weak orthogonal to $\mathcal P$ and $\mathcal Q$  in Example~\ref{ex:orthoqls}. It is also weak orthogonal to any QLS weak orthogonal  to $\mathcal P$ or $\mathcal Q$. To see this consider that any two weak orthogonal QLSs must have columns that are permutations of each other.

This example relies on the \textit{block-like} structure  of the QLSs in question. Any family of MOQLS\ having a similar structure will admit a similar square of vectors. It is unknown whether all QLSs are of this form, but to the authors knowledge none have been found yet that do not have this structure up to equivalence. 

The lower bound for the number of  MOQLS in dimension $n$ must be at least the lower bound for the number of MOLS, more research is required to say any more than that at this stage.
\section{Conclusion}
In our 2015 paper~\cite{mypaper1} the author together with Jamie Vicary introduced the quantum combinatorial objects of quantum Latin squares and gave  a construction of UEBs using them. In this paper we have built upon that work by introducing mutually orthogonal quantum Latin squares which generalise mutually orthogonal Latin squares, which have been used extensively to derive results in quantum information. As an application we have given a construction for mutually unbiased bases in square dimension which gives  MUBs that are inequivalent to those that can be constructed by any known method. There is the potential for improved bounds on maximal families of MUBs in composite dimensions using the main result of this paper. 
\bibliographystyle{eptcs}
\bibliography{orthoQLSPaperArxiv}

\newpage
\appendix
\section{Categorical quantum mechanics}\label{apx:cqm}
The graphical calculus of categorical quantum mechanics gives us a diagrammatic notation through which certain kind of problems are easier. The results of this paper were all achieved using these high level techniques. 

In order to read these diagrams the first thing to understand is that wires represent Hilbert spaces and boxes between wires are linear maps. We will use the convention that diagrams are read from bottom to top. The composition of linear maps $U$ and $V$ is given by vertical composition and the tensor product is given by horizontal composition. We represent $n$-partite states  by triangles with no wires in and $n$ wires out. Scalars are represented by boxes with no wires in or out and can move freely around the diagram.  Adjoints are given by vertical mirror image, so asymmetry in the boxes representing linear maps is introduced to avoid ambiguity. Thus we have the following diagrammatic rendering of $(U \circ V\ket k)\otimes U^\dag\ket l$:
\begin{equation}\label{eq:diagram}
(U \circ V\ket k)\otimes U^\dag\ket l:=
\begin{pic}
\node (f) [morphism,wedge] at (0,0) {$V$};
\node (g) [morphism,wedge, anchor=south] at ([yshift=0.75cm] f.north) {$U$};
\node (b) [state,black,scale=0.5] at (0,-1.05){$k$};
\draw[string] (g.north) to +(0,0.75) node [above] {};
\draw[string] (f.south) to +(0,-0.75) node [below] {};
\draw[string] (f.north) to node [auto] {} (g.south) {};
\end{pic}
\begin{pic}
\node (f) [morphism,hflip,wedge] at (0,0) {$U$};
\node (g) at ([yshift=0.75cm] f.north) {};
\node (b) [state,black,scale=0.5] at (0,-1.45){$l$};
\draw[string] (g.south) to +(0,1.05) node [above] {};
\draw[string] (f.south) to +(0,-1.15) node [below] {};
\draw[string] (f.north) to node [auto] {} (g.south) {};
\end{pic}
\end{equation}
We will  represent quantum Latin squares as linear maps $\tinymult[lsdot]$ and $\tinymult[lssdot]$, these are obtained from QLSs by having the left input wire represent the columns , and the right input wire represent the rows of the QLS indexed by the computational basis states. So the $(i,j)^{th}$ entry $\ket {Q_{ij}}$, of some QLS $\mathcal Q$, is represented by the following diagram:
\begin{equation*}
\ket {Q_{ij}}:=
\begin{aligned}\begin{tikzpicture}          
          \node[state,black,scale=0.5] (0a) at (-1,-0.5) {$i$};
          \node[state,black,scale=0.5] (0c) at (0,-0.5) {$j$};
          \node[lsdot] (7) at (-0.5,0) {};       
          \node (10a) at (-0.5,0.5) {};       
          \draw[string, out=90, in =180] (0a.center) to (7.center);
          \draw[string, out=0, in=90] (7.center) to (0c.center);
          \draw[string] (7.center) to (10a.center);
\end{tikzpicture}\end{aligned} 
\end{equation*}
The final definition we require is that of a classical structure. \textit{Classical structures} are dagger special frobenius algebras. In $\cat{FHilb}$ given an orthonormal basis $\ket i$, classical structures are equivalent to families of linear maps $\mathcal H^{\otimes s}\rightarrow \mathcal H^{\otimes r}$ for varying $s$ and $r$ (possibly zero) of the following form~\cite{cqm2014}:
\begin{equation*}
    \overbrace{\underbrace{
    \begin{pic}
     \node[blackdot,scale=2] (A){};
     \draw[string,out=190,in=90] (A.center) to (-2,-1);
     \draw[string,out=350,in=90] (A.center) to (2,-1);
     \draw[string,out=190,in=90] (A.center) to (-1.7,-1);
     \draw[string,out=350,in=90] (A.center) to (1.7,-1);  
     \draw[string,out=170,in=270] (A.center) to (-2,1);
     \draw[string,out=10,in=270] (A.center) to (2,1);
     \draw[string,out=170,in=270] (A.center) to (-1.7,1);
     \draw[string,out=10,in=270] (A.center) to (1.7,1); 
     \draw[loosely dotted] (-1.65,0.9) to (1.65,0.9); 
     \draw[loosely dotted] (-1.65,-0.9) to (1.65,-0.9);      
    \end{pic}}_{r}}^{s}
    \quad := \quad
     \sum^{n-1}_{i=0}
    \overbrace{\underbrace{
    \begin{pic}
     \node[state,black,scale=0.5] (A)at (-1.75,0.3){$i$};
     \node[state,hflip,black,scale=0.5] (B)at (-2,-0.3){$i$};
     \node[state,black,scale=0.5] (C)at (1.75,0.3){$i$};
     \node[state,hflip,black,scale=0.5] (D)at (2,-0.3){$i$};
     \node[state,black,scale=0.5] (A')at (-1.25,0.3){$i$};
     \node[state,hflip,black,scale=0.5] (B')at (-1.5,-0.3){$i$};
     \node[state,black,scale=0.5] (C')at (1.25,0.3){$i$};
     \node[state,hflip,black,scale=0.5] (D')at (1.5,-0.3){$i$};
     \draw[string] (A) to (-1.75,1);
     \draw[string] (B) to (-2,-1);
     \draw[string] (C) to (1.75,1);
     \draw[string] (D) to (2,-1);  
     \draw[string] (A') to (-1.25,1);
     \draw[string] (B') to (-1.5,-1);
     \draw[string] (C') to (1.25,1);
     \draw[string] (D') to (1.5,-1); 
     \draw[loosely dotted] (-1,0.1) to (1,0.1); 
     \draw[loosely dotted] (-1.25,-0.1) to (1.25,-0.1);          
    \end{pic}}_{r}}^{s}
  \end{equation*}
Classical structures are thus in one to one correspondence with orthonormal bases. It is standard notation to use different colours to represent different bases. Throughout this paper we use the grey classical structure $\tinyspider$ to represent the computational basis. The following theorem gives us a way to rewrite connected diagrams of classical structures.
\begin{theorem}
[Spider merge theorem]\label{thm:sm}Given a family of linear maps $\tinyspider$: $\mathcal H^{\otimes r} \rightarrow \mathcal H^{\otimes s}$ for varying $r,s \in \mathbb{N}$ the following are equivalent:

\begin{itemize}
  \item $\tinyspider$ is a classical structure  \item any connected tensor diagram of the linear maps  with swap maps and identities is equal to the unique linear map from $\mathcal H^{\otimes r}$ to $\mathcal H^{\otimes s}$   
\end{itemize}
e.g.
\begin{equation}\label{eq:sm}
    \overbrace{\underbrace{
    \begin{pic}
     \node[blackdot,scale=2] (A) at (-0.5,0.25){};
     \node[blackdot,scale=2] (B) at (0.5,-0.25){};
     \draw[string,out=190,in=90] (B.center) to (-2,-1);
     \draw[string,out=350,in=90] (B.center) to (2,-1);
     \draw[string,out=190,in=90] (B.center) to (-1.7,-1);
     \draw[string,out=350,in=90] (B.center) to (1.7,-1);  
     \draw[string,out=170,in=270] (A.center) to (-2,1);
     \draw[string,out=10,in=270] (A.center) to (2,1);
     \draw[string,out=170,in=270] (A.center) to (-1.7,1);
     \draw[string,out=10,in=270] (A.center) to (1.7,1); 
     \draw[loosely dotted] (-1.65,0.9) to (1.65,0.9); 
     \draw[loosely dotted] (-1.65,-0.9) to (1.65,-0.9); 
     \draw[dotted] (-0.1,-0.1) to (0.15,0.1);   
     \draw[string,in=170,out=300] (A.center) to (B.center);  
     \draw[string,in=100,out=340] (A.center) to (B.center);   
    \end{pic}}_{r}}^{s}
    \quad = \quad
    \overbrace{\underbrace{
    \begin{pic}
     \node[blackdot,scale=2] (A){};
     \draw[string,out=190,in=90] (A.center) to (-2,-1);
     \draw[string,out=350,in=90] (A.center) to (2,-1);
     \draw[string,out=190,in=90] (A.center) to (-1.7,-1);
     \draw[string,out=350,in=90] (A.center) to (1.7,-1);  
     \draw[string,out=170,in=270] (A.center) to (-2,1);
     \draw[string,out=10,in=270] (A.center) to (2,1);
     \draw[string,out=170,in=270] (A.center) to (-1.7,1);
     \draw[string,out=10,in=270] (A.center) to (1.7,1); 
     \draw[loosely dotted] (-1.65,0.9) to (1.65,0.9); 
     \draw[loosely dotted] (-1.65,-0.9) to (1.65,-0.9);       
    \end{pic}}_{r}}^{s}
\end{equation}
\end{theorem}
Classical structures are also useful for performing linear algebraic operations such as the trace of a linear map, the following diagram shows how this is done:
\begin{equation}\label{eq:trace}
     \text{Trace} (U):= \begin{pic}[xscale=-1]
  \begin{pgfonlayer}{background}
    \node (f) [morphism,wedge] at (0,0) {$U$};
  \end{pgfonlayer}
  \begin{pgfonlayer}{foreground}
    \draw[string] (f.north) to
        +(0,0.75) node [above] {};
    \draw[string] (f.south) to
        +(0,-0.75) node [below] {};
  \end{pgfonlayer}
  \node[blackdot] (A) at (1, 1.75){};
  \node[blackdot] (B) at (1, -1.75){};
  \draw[string] (2,-1) to (2,1);
  \draw[string,out=270,in=180] (0,-1) to (B.center);
  \draw[string,out=270,in=0] (2,-1) to (B.center);
  \draw[string,out=90,in=180] (0,1) to (A.center);
  \draw[string,out=90,in=0] (2,1) to (A.center);
  \end{pic}
    \end{equation}
Classical structures copy the basis states of the corresponding orthonormal basis. 
\begin{equation}\label{eq:copy}
   \begin{pic}
      \node(1)[state,black,scale=0.5] at (0,0){$k$};
      \node(2)[blackdot] at (0,1){};
      \node(3) at (-0.325,2){};
      \node(4) at (0.325,2){};
      \draw[string,out=up,in=down] (1.center) to (2.center){};
      \draw[string,out=left,in=down] (2.center) to (3.center){};
      \draw[string,out=right,in=down] (2.center) to (4.center){};     
   \end{pic}
   \quad = \quad
   \begin{pic}[yscale=-1]
   \node(1)[state,black,scale=0.5]{$k$};
   \draw[string] (1.center) to (0,-2);
   \node(2)[state,black,scale=0.5] at (-0.75,0){$k$};
   \draw[string] (1.center) to (0,-2);
   \draw[string] (2.center) to (-0.75,-2);
   \end{pic}
\end{equation}
If the state $\ket k$ is real valued then the following holds:
\begin{equation}\label{eq:tran}
   \begin{pic}
      \node(1)[state,black,scale=0.5]{$k$};
      \node(2)[blackdot] at (-0.5,1){};
      \node(3) at (-1,-1){};
      \draw[string,out=up,in=right] (1.center) to (2.center){};
      \draw[string,out=left,in=up,looseness=0.65] (2.center) to (3.center){};     
   \end{pic}
   \quad = \quad
   \begin{pic}
   \node(1)[state,black,hflip,scale=0.5]{$k$};
   \draw[string] (1.center) to (0,-2);
   \end{pic}   
\end{equation} 
On the left, the classical structure acts as a transpose which is equal to the adjoint since $\ket k$ is real valued.   
\section{Proof of main theorem}\label{apxprf}
\noindent
\noindent \textbf{Theorem~\ref{thm:main}.}\,
Given two indexed families of $n$ Hadamards $H_k$ and $G_j$ both of size $n \times n$, and  a pair of $n \times n$ weak orthogonal quantum Latin squares $\mathcal P$ and $\mathcal Q$, the bases $B(\mathcal Q, H_k)$ and $B(\mathcal P, G_j)$ are mutually unbiased.
 
\begin{proof}
Let $\mathcal P:=\tinymult[lsdot]$, $\mathcal Q:=\tinymult[lssdot]$ and $\tinyspider$ be the computational basis. By Definition~\ref{def:qlsmeb}
the $(i,j)^{th}$ state of the basis $\mathcal A$ and the $(p,q)^{th}$ state of the basis $\mathcal B$ are as follows:
\begin{align*}
A_{ij} &= \frac{1}{\sqrt{n}}\sum_{k=0}^{n-1}\ket k \otimes \ketbra{P_{kj}}{k}H_j\ket i\\
B_{pq} &= \frac{1}{\sqrt{n}}\sum_{s=0}^{n-1}\ket s \otimes \ketbra{Q_{sq}}{s}G_q\ket p
\end{align*}
Graphically they are as follows:
\begin{equation}\label{eq:AB}A_{ij}:=
\frac{1}{\sqrt{n}} \begin{pic}
          \node (in) at (-0.5,4) {};         
          \node (H)[morphism,wedge,scale=0.5] at(0.25,1.5) {$H_j$};
          \node (b1)[blackdot,scale=1] at (0.25,2) {};
          \node (b2)[lsdot,scale=1.2] at (1,3) {};          
          \node (i)[state,black,scale=0.5] at(0.25,1) {$i$};
          \node (0b) at (0.5,2) {};
          \node (j)[state,black,scale=0.5] at(1.5,2) {$j$};
          \node (out) at (1,4) {}; 
          \draw[string,out=270,in=180,looseness=0.75] (in.center) to (b1.center);           \draw (H.north) to (b1.center);
          \draw[string,out=0,in=180] (b1.center) to (b2.center);
          \draw[string,out=90,in=0] (j.center) to (b2.center);
          \draw[string] (b2.center) to (out.center);
          \draw[string] (i) to (H.south);
\end{pic} \qquad \qquad B_{pq}:=
\frac{1}{\sqrt{n}}\begin{pic}
          \node (in) at (-0.5,4) {};         
          \node (H)[morphism,wedge,scale=0.5] at(0.25,1.5) {$G_q$};
          \node (b1)[blackdot,scale=1] at (0.25,2) {};
          \node (b2)[lssdot,scale=1.2] at (1,3) {};          
          \node (i)[state,black,scale=0.5] at(0.25,1) {$p$};
          \node (0b) at (0.5,2) {};
          \node (j)[state,black,scale=0.5] at(1.5,2) {$q$};
          \node (out) at (1,4) {}; 
          \draw[string,out=270,in=180,looseness=0.75] (in.center) to (b1.center);           \draw (H.north) to (b1.center);
          \draw[string,out=0,in=180] (b1.center) to (b2.center);
          \draw[string,out=90,in=0] (j.center) to (b2.center);
          \draw[string] (b2.center) to (out.center);
          \draw[string] (i) to (H.south);
          \end{pic}
\end{equation} 
$\mathcal P$ and $\mathcal Q$ are weak orthogonal so by Equation~\eqref{eq:lols}, $f$ defined as follows is a function on computational basis states:
\begin{equation}\label{eq:lolss}
\begin{pic}
   \node[whitedot,scale=2](p) {$f$};
   \draw[string] (p.center) to (0,1.2);
   \draw[string] (0.475,0) to (0.475,-1.2);
   \draw[string] (-0.475,0) to (-0.475,-1.2);
\end{pic}
\quad := \quad
\begin{aligned}\begin{pic}
   \node[lsdot] (ls) at (0,-0.25){};  
   \node[lssdot] (lss) at (0,1){};
   \node (b1) at(-0.75,0.375){};
   \node[blackdot,scale=0.8] (b3) at (0.4,1.3){};
   \node[blackdot,scale=0.8] (b5) at (-0.375,1.25){};
   \node[blackdot,scale=0.8] (b4) at (-0.375,-0.5){};
   \draw[string,out=left,in=up] (b5.center) to (b1.center);
   \draw[string, in=right,out=left] (ls.center) to (b4.center);
   \draw[string,out=left,in=down] (b4.center) to (b1.center);
   \draw[string, in=down,out=up] (ls.center) to (lss.center);    
   \draw[string, in=right,out=left] (lss.center) to (b5.center);
   \draw[string, in=left,out=right,looseness=1.5] (lss.center) to (b3.center);   \node (1) at (0.25,-0.8){};
   \draw[string,in=right,out=up] (1.center) to (ls.center);  
   \node (2) at (-0.375,1.8){};
   \draw[string] (2.center) to (b5.center); 
   \node (3) at (1,1.8){};
   \node (4) at (1.25,-0.8){};
   \draw[string,in=right,out=up] (4.center) to (b3.center);
\end{pic}\end{aligned} 
\end{equation}
Since $f$ is a function on basis states, $f(\ket{j,q})$ is a computational basis state, say $\ket t$ i.e.\begin{equation}\label{eq:t}
\begin{pic}
   \node[whitedot,scale=2](p) {$f$};
   \node[state,black,scale=0.5] at (0.475,-1.2){$q$};
   \node[state,black,scale=0.5] at (-0.475,-1.2){$j$};
   \draw[string] (p.center) to (0,1.2);
   \draw[string] (0.475,0) to (0.475,-1.2);
   \draw[string] (-0.475,0) to (-0.475,-1.2);
\end{pic}
=
\begin{pic}
   \node(t)[state,black,scale=0.5]{$t$};
   \draw[string] (t) to +(0,2.4);
\end{pic}
\end{equation}
We are now ready to show that $\mathcal A$ and $\mathcal B$ are mutually unbiased.
\begin{align*}
|\braket{B_{pq}}{A_{ij}}|^2&\super{\eqref{eq:AB}}= \left |
\frac{1}{n}
\begin{pic}
          \node (in) at (-0.5,0) {};         
          \node (H)[morphism,wedge,scale=0.5] at(0.25,-2) {$H_j$};
          \node (b1)[blackdot,scale=1] at (0.25,-1.5) {};
          \node (b2)[lsdot,scale=1.2] at (1,-0.5) {};          
          \node (i)[state,black,scale=0.5] at(0.25,-2.25) {$i$};
          \node (0b) at (0.5,-1.5) {};
          \node (j)[state,black,scale=0.5] at(1.5,-1.5) {$j$};
          \node (out) at (1,0) {}; 
          \draw[string,out=270,in=180,looseness=0.75] (in.center) to (b1.center);           \draw (H.north) to (b1.center);
          \draw[string,out=0,in=180] (b1.center) to (b2.center);
          \draw[string,out=90,in=0] (j.center) to (b2.center);
          \draw[string] (b2.center) to (out.center);
          \draw[string] (i) to (H.south);
         \node (in') at (-0.5,0) {};         
          \node (H')[morphism,wedge,,hflip,scale=0.5] at(0.25,2) {$G_q$};
          \node (b1')[blackdot,scale=1] at (0.25,1.5) {};
          \node (b2')[lssdot,scale=1.2] at (1,0.5) {};          
          \node (i')[state,black,hflip,scale=0.5] at(0.25,2.25) {$p$};
          \node (0b') at (0.5,2) {};
          \node (j')[state,black,hflip,scale=0.5] at(1.5,1.5) {$q$};
          \node (out') at (1,0) {}; 
          \draw[string,out=90,in=180,looseness=0.75] (in'.center) to (b1'.center);           \draw[string] (H'.south) to (b1'.center);
          \draw[string,out=0,in=180] (b1'.center) to (b2'.center);
          \draw[string,out=270,in=0] (j'.center) to (b2'.center);
          \draw[string] (b2'.center) to (out'.center);
          \draw[string] (i') to (H'.north);
          \end{pic}\right |^2
\quad \super{\eqref{eq:tran}}= 
\frac{1}{n^2}\left |
\begin{pic}
          \node (in) at (-0.5,0) {};         
          \node (H)[morphism,wedge,scale=0.5] at(0.25,-2) {$H_j$};
          \node (b1)[blackdot,scale=1] at (0.25,-1.5) {};
          \node (b2)[lsdot,scale=1.2] at (1,-0.5) {};          
          \node (i)[state,black,scale=0.5] at(0.25,-2.25) {$i$};
          \node (0b) at (0.5,-1.5) {};
          \node (j)[state,black,scale=0.5] at(1.5,-2) {$j$};
          \node (out) at (1,0) {}; 
          \draw[string,out=270,in=180,looseness=0.75] (in.center) to (b1.center);           \draw (H.north) to (b1.center);
          \draw[string,out=0,in=180] (b1.center) to (b2.center);
          \draw[string,out=90,in=0] (j.center) to (b2.center);
          \draw[string] (b2.center) to (out.center);
          \draw[string] (i) to (H.south);
          \node (in') at (-0.5,0) {};         
          \node (H')[morphism,wedge,,hflip,scale=0.5] at(0.25,2) {$G_q$};
          \node (b1')[blackdot,scale=1] at (0.25,1.5) {};
          \node (b2')[lssdot,scale=1.2] at (1,0.5) {};          
          \node (i')[state,black,hflip,scale=0.5] at(0.25,2.25) {$p$};
          \node (0b') at (0.5,2) {};
          \node (j')[state,black,scale=0.5] at(2.5,-2) {$q$};
          \node (out') at (1,0) {};
          \node (b) [blackdot] at (1.875,1.25){}; 
          \draw[string,out=90,in=180,looseness=0.75] (in'.center) to (b1'.center);           \draw[string] (H'.south) to (b1'.center);
          \draw[string,out=0,in=180] (b1'.center) to (b2'.center);
          \draw[string,out=90,in=0,looseness=0.6] (j'.center) to (b.center);
          \draw[string,out=180,in=0,looseness=1.2] (b.center) to (b2'.center);                   \draw[string] (b2'.center) to (out'.center);
          \draw[string] (i') to (H'.north);
          \end{pic}\right |^2 
\super{\eqref{eq:t}}=
\frac{1}{n^2}\left | 
\begin{pic}
          \node (in) at (-0.5,0) {};         
          \node (H)[morphism,wedge,scale=0.5] at(-0.75,-2) {$H_j$};  
          \node (b1)[blackdot,scale=1] at (0.25,-1) {};
          \node (b2)[lsdot,scale=1.2] at (1,-0.5) {};          
          \node (i)[state,black,scale=0.5] at(-0.75,-2.25) {$i$};
          \node (0b) at (0.5,-1.5) {};
          \node (j)[state,black,scale=0.5] at(1.5,-2) {$j$};
          \node (out) at (1,0) {};
          \node (d)[blackdot] at (-0.35,2){}; 
          \node (c) [blackdot] at (0.25,1.5){};
          \draw[string,out=270,in=180,looseness=0.75] (in.center) to (b1.center);           \draw[string,in=left,out=up,looseness=0.35] (H.north) to (d.center);
          \draw[string,in=left,out=right] (d.center) to (c.center);
          \draw[string,out=0,in=180] (b1.center) to (b2.center);
          \draw[string,out=90,in=0] (j.center) to (b2.center);
          \draw[string] (b2.center) to (out.center);
          \draw[string] (i) to (H.south);
         \node (in') at (-0.5,0) {};         
          \node (H')[morphism,wedge,,hflip,scale=0.5] at(0.75,2) {$G_q$};
          \node (b1')[blackdot,scale=1] at (0.25,1) {};
          \node (b2')[lssdot,scale=1.2] at (1,0.5) {};          
          \node (i')[state,black,hflip,scale=0.5] at(0.75,2.25) {$p$};
          \node (0b') at (0.5,2) {};
          \node (j')[state,black,scale=0.5] at(2.5,-2) {$q$};
          \node (out') at (1,0) {};
          \node (b) [blackdot] at (1.875,1.25){}; 
          \draw[string,out=90,in=180,looseness=0.75] (in'.center) to (b1'.center);           \draw[string,in=left,out=down] (H'.south) to (c.center);
          \draw[string] (b1'.center) to (c.center);
          \draw[string,out=0,in=180] (b1'.center) to (b2'.center);
          \draw[string,out=90,in=0,looseness=0.6] (j'.center) to (b.center);
          \draw[string,out=180,in=0,looseness=1.2] (b.center) to (b2'.center);            \draw[string] (b2'.center) to (out'.center);
          \draw[string] (i') to (H'.north);
          \end{pic}\right |^2 
\\&\super{\eqref{eq:lolss}}=\quad
\frac{1}{n^2}\left |
\begin{pic}
          \node (H1)[morphism,wedge,scale=0.5] at(-0.75,-2) {$H_j$};
          \node (f)[whitedot,scale=2] at (1,-0.25) {$f$};
          \node (f1) at (1,0.75) {};
          \node (f2) at (0.525,-1.25) {};
          \node (f3) at (1.475,-1.25) {}; 
          \node (f2a) at (0.525,-0.5) {};
          \node (f3a) at (1.475,-0.5) {};                    
          \node (i)[state,black,scale=0.5] at(-0.75,-2.25) {$i$};
          \node (j)[state,black,scale=0.5] at(1.25,-2) {$j$};
          \node (B1)[blackdot] at (-0.35,2){}; 
          \node (H2)[morphism,wedge,,hflip,scale=0.5] at(0.75,2) {$G_q$};
          \node (B2) [blackdot] at (0.25,1.5){};
          \node (p)[state,black,hflip,scale=0.5] at(0.75,2.25) {$p$};
          \node (q)[state,black,scale=0.5] at(2.25,-2) {$q$};
          \draw[string,out=up,in=down] (q.center) to (f3.center);
          \draw[string,out=up,in=down] (j.center) to (f2.center);
          \draw[string] (f3.center) to (f3a.center);                               \draw[string] (f2.center) to (f2a.center);
          \draw[string] (f.center) to (f1.center);
          \draw[string,out=up,in=down] (f1.center) to (B2.center);
          \draw[string,out=right,in=down] (B2.center) to (H2.south);
          \draw[string,out=left,in=right] (B2.center) to (B1.center);
          \draw[string] (H2.north) to (p.center);
          \draw[string,out=left,in=up,looseness=0.35] (B1.center) to (H1.north);
          \draw[string] (H1.south) to (i.center);
\end{pic}\right |^2                     
\super{\eqref{eq:t}}=\quad
\frac{1}{n^2}\left |
\begin{pic}
          \node (H1)[morphism,wedge,scale=0.5] at(-0.75,-2) {$H_j$};
          \node (f)[state,black,scale=0.5] at (1,-0.25) {$t$};
          \node (f1) at (1,0.75) {};                       
          \node (i)[state,black,scale=0.5] at(-0.75,-2.25) {$i$};
          \node (B1)[blackdot] at (-0.35,2){}; 
          \node (H2)[morphism,wedge,,hflip,scale=0.5] at(0.75,2) {$G_q$};
          \node (B2) [blackdot] at (0.25,1.5){};
          \node (p)[state,black,hflip,scale=0.5] at(0.75,2.25) {$p$};
          \draw[string] (f.center) to (f1.center);
          \draw[string,out=up,in=down] (f1.center) to (B2.center);
          \draw[string,out=right,in=down] (B2.center) to (H2.south);
          \draw[string,out=left,in=right] (B2.center) to (B1.center);
          \draw[string] (H2.north) to (p.center);
          \draw[string,out=left,in=up,looseness=0.35] (B1.center) to (H1.north);
          \draw[string] (H1.south) to (i.center);
\end{pic}\right |^2  
\super{\eqref{eq:copy}}=\quad
\frac{1}{n^2}\left |
\begin{pic}
          \node (H1)[morphism,wedge,scale=0.5] at(-0.75,-2) {$H_j$};
          \node (f)[state,black,scale=0.5] at (0.75,1.75) {$t$};
          \node (f1) at (1,0.75) {};                       
          \node (i)[state,black,scale=0.5] at(-0.75,-2.25) {$i$};
          \node (B1)[blackdot] at (-0.35,2){}; 
          \node (H2)[morphism,wedge,,hflip,scale=0.5] at(0.75,2) {$G_q$};
          \node (B2) [state,black,scale=0.5] at (0,1.5){t};
          \node (p)[state,black,hflip,scale=0.5] at(0.75,2.25) {$p$};
          \draw[string] (f.center) to (H2.south);
          \draw[string,out=up,in=right] (B2.center) to (B1.center);
          \draw[string] (H2.north) to (p.center);
          \draw[string,out=left,in=up,looseness=0.35] (B1.center) to (H1.north);
          \draw[string] (H1.south) to (i.center);
\end{pic}\right |^2
\\&\super{\eqref{eq:tran}}=\quad
\frac{1}{n^2}\left |
\begin{pic}
          \node (H1)[morphism,wedge,scale=0.5] at(-0.75,0) {$H_j$};
          \node (f)[state,black,scale=0.5] at (0,-0.25) {$t$};
          \node (i)[state,black,scale=0.5] at(-0.75,-0.25) {$i$};
          \node (H2)[morphism,wedge,,hflip,scale=0.5] at(0,0) {$G_q$};
          \node (B2) [state,black,hflip,scale=0.5] at (-0.75,0.25){$t$};
          \node (p)[state,black,hflip,scale=0.5] at(0,0.25) {$p$};
          \draw[string] (f.center) to (H2.south);
          \draw[string] (H2.north) to (p.center);
          \draw[string] (H1.south) to (i.center);
          \draw[string] (H1.north) to (B2.center);
\end{pic}\right |^2  
\super{\eqref{eq:diagram}}=\frac{1}{n^2}|(H_j^\pdag)_{it}(G_q^\dag)_{tp}|^2\super{\eqref{had1}}=\frac{1}{n^2}1^2=\frac{1}{n^2} \end{align*}
\end{proof}

\section{Quantum Latin square $9 \times 9$ example  MUB}\label{apx:6561}
We now give a sample of the $81$ states of basis $\mathcal A$ and the $81$ states of basis $\mathcal B$ from Example~\ref{ex:qlsmub}, with some calculations of their inner products showing mutual unbiasedness. We give everything in terms of the computational basis states $\ket{i,j}$ such that $i ,j\in \{0,...,n-1\}$. And we define the scalar $\omega:=e^{2\pi i/3}$. Here are some states from $\mathcal A$ and $\mathcal B$:
\begin{align*}
\mathcal A_{74}&= \frac{1}{3}(\ket{0,8}+\omega^2\ket{1,7}+\omega\ket{2,6}+\omega^2\ket{3,2}+\omega\ket{4,1}+\ket{5,0}+\omega\ket{6,5}+\ket{7,4}+\omega^2\ket{8,3})\\
\mathcal A_{46}&=\frac{1}{3}(\frac{\omega}{\sqrt{3}}\ket{0,3}+\frac{\omega^2}{\sqrt{3}}\ket{0,4}+\frac{i}{\sqrt{3}}\ket{0,5}-\omega\sqrt{\frac{2}{7}}\ket{1,3}-\frac{i\omega^2}{\sqrt{14}}\ket{1,4}+\frac{3}{\sqrt{14}}\ket{1,5}+i\omega\sqrt{\frac{2}{3}}\ket{2,3}\\&-\frac{\omega^2}{\sqrt{6}}\ket{2,4}+\frac{i}{\sqrt{6}}\ket{2,5}+\omega^2\ket{3,6}+\omega\ket{4,8}+\ket{5,7}+\frac{1}{\sqrt{3}}\ket{6,0}+\frac{\omega}{\sqrt{3}}\ket{6,1}+\frac{\omega^2}{\sqrt{3}}\ket{6,2}
\\&+\frac{1}{\sqrt{3}}\ket{7,0}+\frac{1}{\sqrt{3}}\ket{7,1}+\frac{1}{\sqrt{3}}\ket{7,2}+\frac{1}{\sqrt{3}}\ket{8,0}+\frac{\omega^2}{\sqrt{3}}\ket{8,1}+\frac{\omega}{\sqrt{3}}\ket{8,2})\\
\mathcal B_{38}&=\frac{1}{3}(\ket{0,7}+\ket{1,8}+\ket{2,6}+\omega\ket{3,4}+\omega\ket{4,5}+\omega\ket{5,3}+\frac{\omega^2}{\sqrt{3}}\ket{6,0}+\frac{1}{\sqrt{3}}\ket{6,1}+\frac{\omega}{\sqrt{3}}\ket{6,2}\\
&+\frac{\omega^2}{\sqrt{3}}\ket{7,0}+\frac{\omega}{\sqrt{3}}\ket{7,1}+\frac{1}{\sqrt{3}}\ket{7,2}+\frac{\omega^2}{\sqrt{3}}\ket{8,0}+\frac{\omega^2}{\sqrt{3}}\ket{8,1}+\frac{\omega^2}{\sqrt{3}}\ket{8,2})\\
\mathcal B_{03}&=\frac{1}{3}(\ket{0,1}+\ket{1,2}+ \ket{2,0}+\ket{3,7}+ \ket{4,8}+\ket{5,6}+ \ket{6,4}+\ket{7,5}+\ket{8,3})  
\end{align*} 
Here are some  calculations for mutual unbiasedness. Note that they all equal $\frac{1}{81}$ as required:
\begin{align*}
|\braket{\mathcal A_{74}}{B_{38}}|^2&=|\frac{1}{9}\omega|^2=\frac{1}{81}\\
|\braket{\mathcal A_{74}}{\mathcal B_{03}}|^2&=|\frac{1}{9}\omega^2|^2=\frac{1}{81}\\
|\braket{\mathcal A_{46}}{\mathcal B_{38}}|^2&=|\frac{1}{9}\left [\frac{1}{3}(\omega^2+\omega+1)+\frac{1}{3}(\omega^2+\omega+1)+\frac{1}{3}(\omega^2+\omega+1)\right ]|^2=\frac{1}{81}\\
|\braket{\mathcal A_{46}}{\mathcal B_{03}}|^2&=|\frac{1}{9}\omega|^2=\frac{1}{81}
\end{align*}
\end{document}

%% file: book-header.tex
\usepackage{docmute}
\usepackage{hyperref}
\usepackage{enumitem,multirow}
\usepackage{amsmath,amsfonts,amssymb,mathtools,xspace}
\usepackage{amsthm}
\usepackage{tikz}
\usepackage{array}
\usetikzlibrary{decorations.pathreplacing}
\usetikzlibrary{decorations.markings}
\usetikzlibrary{calc}
\usepackage{etoolbox}
\usepackage{datetime}

\tikzstyle{dc}   = [circle, minimum width=8pt, draw, inner sep=0pt, path picture={\draw (path picture bounding box.south east) -- (path picture bounding box.north west) (path picture bounding box.south west) -- (path picture bounding box.north east);}]
\tikzstyle{dp}   = [circle, minimum width=8pt, draw, inner sep=0pt, path picture={\draw (path picture bounding box.west) -- (path picture bounding box.north) (path picture bounding box.south west) -- (path picture bounding box.north) (path picture bounding box.south west) -- (path picture bounding box.north east) (path picture bounding box.north east) -- (path picture bounding box.south) (path picture bounding box.south) -- (path picture bounding box.east);}]
\tikzstyle{ls} = [dc,scale=0.65]
\tikzstyle{ls'} = [dp,scale=1.3]
\tikzstyle{dpt}  = [circle, minimum width=8pt, draw, inner sep=0pt, path picture={\draw (path picture bounding box.south) -- (path picture bounding box.north) (path picture bounding box.west) -- (path picture bounding box.east);}]
\tikzstyle{agg} = [dpt,scale=0.65]

\renewenvironment{pic}[1][]
{\begin{aligned}\begin{tikzpicture}[#1]}
{\end{tikzpicture}\end{aligned}}

\makeatletter
\def\calign@preamble{%
   &\hfil\strut@
    \setboxz@h{\@lign$\m@th\displaystyle{##}$}%
    \ifmeasuring@\savefieldlength@\fi
    \set@field
    \hfil
    \tabskip\alignsep@
}
\let\cmeasure@\measure@
\patchcmd\cmeasure@{\divide\@tempcntb\tw@}{}{}{}
\patchcmd\cmeasure@{\divide\@tempcntb\tw@}{}{}{}
\patchcmd\cmeasure@{\ifodd\maxfields@
  \global\advance\maxfields@\@ne
  \fi}{}{}{}    

\makeatother


\pgfdeclarelayer{foreground}
\pgfdeclarelayer{background}
\pgfsetlayers{main,foreground,background}

\makeatletter
\pgfkeys{%
  /tikz/on layer/.code={
    \pgfonlayer{#1}\begingroup
    \aftergroup\endpgfonlayer
    \aftergroup\endgroup
  },
  /tikz/node on layer/.code={
    \gdef\node@@on@layer{%
      \setbox\tikz@tempbox=\hbox\bgroup\pgfonlayer{#1}\unhbox\tikz@tempbox\endpgfonlayer\egroup}
    \aftergroup\node@on@layer
  },
  /tikz/end node on layer/.code={
    \endpgfonlayer\endgroup\endgroup
  }
}
\def\node@on@layer{\aftergroup\node@@on@layer}
\makeatother\def\thickness{0.7pt}
\tikzstyle{oldmorphism}=[minimum width=30pt, minimum height=16pt, draw, font=\small, inner sep=0pt, fill=white, line width=\thickness]
\tikzstyle{cross}=[preaction={draw=white, -, line width=10pt}]
\tikzstyle{braid}=[double=black, line width=3*\thickness, double distance=\thickness, white]
\tikzstyle{string}=[line width=\thickness]
\tikzstyle{scalar}=[circle, inner sep=0pt, minimum width=15pt, draw, line width=\thickness]
\tikzstyle{dot}=[circle, draw=black, fill=black!25, inner sep=.4ex, line width=\thickness, node on layer=foreground]
\tikzstyle{blackdot}=[circle, draw=black, fill=black!35, inner sep=.5ex, line width=\thickness, node on layer=foreground]
\tikzstyle{whitedot}=[circle, draw=black, fill=white, inner sep=.5ex, line width=\thickness, node on layer=foreground]
\tikzstyle{reddot}=[circle, draw=black, fill=red, inner sep=.4ex, line width=\thickness, node on layer=foreground]
\tikzstyle{bluedot}=[circle, draw=black, fill=blue, inner sep=.4ex, line width=\thickness, node on layer=foreground]
\tikzstyle{greendot}=[circle, draw=black, fill=green, inner sep=.4ex, line width=\thickness, node on layer=foreground]
\tikzstyle{lsdot}=[circle, draw=black, fill=white, inner sep=.4ex, line width=\thickness, node on layer=foreground, path picture={\draw (path picture bounding box.south east) -- (path picture bounding box.north west) (path picture bounding box.south west) -- (path picture bounding box.north east);}]
\tikzstyle{lssdot}=[circle, draw=black, fill=white, inner sep=.4ex, line width=\thickness, node on layer=foreground, path picture={\draw (path picture bounding box.south) -- (path picture bounding box.north) (path picture bounding box.west) -- (path picture bounding box.east);}]
\tikzstyle{mixedmorphism}=[morphism, minimum width=30pt, minimum height=16pt, draw, font=\small, inner sep=0pt, fill=white, line width=\thickness,rounded corners=1ex]
\tikzstyle{thick}=[line width=\thickness]
\tikzstyle{tiny}=[font=\tiny]

\tikzset{arrow/.style={decoration={
    markings,
    mark=at position #1 with \arrow{thickarrow}},
    postaction=decorate}
}
\tikzset{reverse arrow/.style={decoration={
    markings,
    mark=at position #1 with \arrow{reversethickarrow}},
    postaction=decorate}
}

\newif\ifblack\pgfkeys{/tikz/black/.is if=black}
\newif\ifwedge\pgfkeys{/tikz/wedge/.is if=wedge}
\newif\ifvflip\pgfkeys{/tikz/vflip/.is if=vflip}
\newif\ifhflip\pgfkeys{/tikz/hflip/.is if=hflip}
\newif\ifhvflip\pgfkeys{/tikz/hvflip/.is if=hvflip}
\newif\ifconnectnw\pgfkeys{/tikz/connect nw/.is if=connectnw}
\newif\ifconnectne\pgfkeys{/tikz/connect ne/.is if=connectne}
\newif\ifconnectsw\pgfkeys{/tikz/connect sw/.is if=connectsw}
\newif\ifconnectse\pgfkeys{/tikz/connect se/.is if=connectse}
\newif\ifconnectn\pgfkeys{/tikz/connect n/.is if=connectn}
\newif\ifconnects\pgfkeys{/tikz/connect s/.is if=connects}
\newif\ifconnectnwf\pgfkeys{/tikz/connect nw >/.is if=connectnwf}
\newif\ifconnectnef\pgfkeys{/tikz/connect ne >/.is if=connectnef}
\newif\ifconnectswf\pgfkeys{/tikz/connect sw >/.is if=connectswf}
\newif\ifconnectsef\pgfkeys{/tikz/connect se >/.is if=connectsef}
\newif\ifconnectnf\pgfkeys{/tikz/connect n >/.is if=connectnf}
\newif\ifconnectsf\pgfkeys{/tikz/connect s >/.is if=connectsf}
\newif\ifconnectnwr\pgfkeys{/tikz/connect nw </.is if=connectnwr}
\newif\ifconnectner\pgfkeys{/tikz/connect ne </.is if=connectner}
\newif\ifconnectswr\pgfkeys{/tikz/connect sw </.is if=connectswr}
\newif\ifconnectser\pgfkeys{/tikz/connect se </.is if=connectser}
\newif\ifconnectnr\pgfkeys{/tikz/connect n </.is if=connectnr}
\newif\ifconnectsr\pgfkeys{/tikz/connect s </.is if=connectsr}
\tikzset{keylengthnw/.initial=\connectheight}
\tikzset{keylengthn/.initial =\connectheight}
\tikzset{keylengthne/.initial=\connectheight}
\tikzset{keylengthsw/.initial=\connectheight}
\tikzset{keylengths/.initial =\connectheight}
\tikzset{keylengthse/.initial=\connectheight}
\tikzset{connect nw length/.style={connect nw=true, keylengthnw={#1}}}
\tikzset{connect n length/.style ={connect n =true, keylengthn ={#1}}}
\tikzset{connect ne length/.style={connect ne=true, keylengthne={#1}}}
\tikzset{connect sw length/.style={connect sw=true, keylengthsw={#1}}}
\tikzset{connect s length/.style ={connect s =true, keylengths ={#1}}}
\tikzset{connect se length/.style={connect se=true, keylengthse={#1}}}
\tikzset{connect nw < length/.style={connect nw <=true, keylengthnw={#1}}}
\tikzset{connect n < length/.style ={connect n <=true,  keylengthn ={#1}}}
\tikzset{connect ne < length/.style={connect ne <=true, keylengthne={#1}}}
\tikzset{connect sw < length/.style={connect sw <=true, keylengthnw={#1}}}
\tikzset{connect s < length/.style ={connect s <=true,  keylengths ={#1}}}
\tikzset{connect se < length/.style={connect se <=true, keylengthse={#1}}}
\tikzset{connect nw > length/.style={connect nw >=true, keylengthnw={#1}}}
\tikzset{connect n > length/.style ={connect n >=true,  keylengthn ={#1}}}
\tikzset{connect ne > length/.style={connect ne >=true, keylengthne={#1}}}
\tikzset{connect sw > length/.style={connect sw >=true, keylengthsw={#1}}}
\tikzset{connect s > length/.style ={connect s >=true,  keylengths ={#1}}}
\tikzset{connect se > length/.style={connect se >=true, keylengthse={#1}}}

\newlength\morphismheight
\newlength\minimummorphismwidth
\newlength\stateheight
\newlength\minimumstatewidth
\newlength\connectheight
\tikzset{width/.initial=\minimummorphismwidth}

\makeatletter
\pgfarrowsdeclare{thickarrow}{thickarrow}
{
  \pgfutil@tempdima=-0.84pt%
  \advance\pgfutil@tempdima by-1.3\pgflinewidth%
  \pgfutil@tempdimb=-1.7pt%
  \advance\pgfutil@tempdimb by.625\pgflinewidth%
  \pgfarrowsleftextend{+\pgfutil@tempdima}
  \pgfarrowsrightextend{+\pgfutil@tempdimb}
}
{
  \pgfmathparse{\pgfgetarrowoptions{thickarrow}}%
  \pgfsetlinewidth{1.25 pt}
  \pgfutil@tempdima=0.28pt%
  \advance\pgfutil@tempdima by.3\pgflinewidth%
  \pgfsetlinewidth{0.8\pgflinewidth}
  \pgfsetdash{}{+0pt}
  \pgfsetroundcap
  \pgfsetroundjoin
  \pgfpathmoveto{\pgfqpoint{-3\pgfutil@tempdima}{4\pgfutil@tempdima}}
  \pgfpathcurveto
  {\pgfqpoint{-2.75\pgfutil@tempdima}{2.5\pgfutil@tempdima}}
  {\pgfqpoint{0pt}{0.25\pgfutil@tempdima}}
  {\pgfqpoint{0.75\pgfutil@tempdima}{0pt}}
  \pgfpathcurveto
  {\pgfqpoint{0pt}{-0.25\pgfutil@tempdima}}
  {\pgfqpoint{-2.75\pgfutil@tempdima}{-2.5\pgfutil@tempdima}}
  {\pgfqpoint{-3\pgfutil@tempdima}{-4\pgfutil@tempdima}}
  \pgfusepathqstroke
}
\pgfarrowsdeclare{reversethickarrow}{reversethickarrow}
{
  \pgfutil@tempdima=-0.84pt%
  \advance\pgfutil@tempdima by-1.3\pgflinewidth%
  \pgfutil@tempdimb=0.2pt%
  \advance\pgfutil@tempdimb by.625\pgflinewidth%
  \pgfarrowsleftextend{+\pgfutil@tempdima}
  \pgfarrowsrightextend{+\pgfutil@tempdimb}
}
{
  \pgftransformxscale{-1}
  \pgfmathparse{\pgfgetarrowoptions{thickarrow}}%
  \ifpgfmathunitsdeclared%
    \pgfmathparse{\pgfmathresult pt}%
  \else%
    \pgfmathparse{\pgfmathresult*\pgflinewidth}%
  \fi%
  \let\thickness=\pgfmathresult
  \pgfsetlinewidth{1.25 pt}
  \pgfutil@tempdima=0.28pt%
  \advance\pgfutil@tempdima by.3\pgflinewidth%
  \pgfsetlinewidth{0.8\pgflinewidth}
  \pgfsetdash{}{+0pt}
  \pgfsetroundcap
  \pgfsetroundjoin
  \pgfpathmoveto{\pgfqpoint{-3\pgfutil@tempdima}{4\pgfutil@tempdima}}
  \pgfpathcurveto
  {\pgfqpoint{-2.75\pgfutil@tempdima}{2.5\pgfutil@tempdima}}
  {\pgfqpoint{0pt}{0.25\pgfutil@tempdima}}
  {\pgfqpoint{0.75\pgfutil@tempdima}{0pt}}
  \pgfpathcurveto
  {\pgfqpoint{0pt}{-0.25\pgfutil@tempdima}}
  {\pgfqpoint{-2.75\pgfutil@tempdima}{-2.5\pgfutil@tempdima}}
  {\pgfqpoint{-3\pgfutil@tempdima}{-4\pgfutil@tempdima}}
  \pgfusepathqstroke
}
\makeatother

\makeatletter
\pgfdeclareshape{ground}
{
    \savedanchor\centerpoint
    {
        \pgf@x=0pt
        \pgf@y=0pt
    }
    \anchor{center}{\centerpoint}
    \anchorborder{\centerpoint}
    \anchor{north}
    {
        \pgf@x=0pt
        \pgf@y=0.5*0.33*\stateheight
    }
    \anchor{south}
    {
        \pgf@x=0pt
        \pgf@y=0pt
    }
    \saveddimen\overallwidth
    {
        \pgfkeysgetvalue{/pgf/minimum width}{\minwidth}
        \pgf@x=\minimumstatewidth
        \ifdim\pgf@x<\minwidth
            \pgf@x=\minwidth
        \fi
    }
    \backgroundpath
    {
        \begin{pgfonlayer}{foreground}
        \pgfsetstrokecolor{black}
        \pgfsetcolor{gray}
        \pgfsetlinewidth{1.25pt}
        \ifhflip
            \pgftransformyscale{-1}
        \fi
        \pgftransformscale{0.5}
        \pgfpathmoveto{\pgfpoint{-0.5*\overallwidth}{0}}
        \pgfpathlineto{\pgfpoint{0.5*\overallwidth}{0}}
        \pgfpathmoveto{\pgfpoint{-0.33*\overallwidth}{0.33*\stateheight}}
        \pgfpathlineto{\pgfpoint{0.33*\overallwidth}{0.33*\stateheight}}
        \pgfpathmoveto{\pgfpoint{-0.16*\overallwidth}{0.66*\stateheight}}
        \pgfpathlineto{\pgfpoint{0.16*\overallwidth}{0.66*\stateheight}}
        \pgfpathmoveto{\pgfpoint{-0.02*\overallwidth}{\stateheight}}
        \pgfpathlineto{\pgfpoint{0.02*\overallwidth}{\stateheight}}
        \end{pgfonlayer}
    }
}
\tikzset{forward arrow style/.style={every to/.style, decoration={
    markings,
    mark=at position 0.5 with \arrow{thickarrow}},
    postaction=decorate}}
\tikzset{reverse arrow style/.style={every to/.style, decoration={
    markings,
    mark=at position 0.5 with \arrow{reversethickarrow}},
    postaction=decorate}}
\pgfdeclareshape{morphism}
{
    \savedanchor\centerpoint
    {
        \pgf@x=0pt
        \pgf@y=0pt
    }
    \anchor{center}{\centerpoint}
    \anchorborder{\centerpoint}
    \saveddimen\savedlengthnw
    {
        \pgfkeysgetvalue{/tikz/keylengthnw}{\len}
        \pgf@x=\len
    }
    \saveddimen\savedlengthn
    {
        \pgfkeysgetvalue{/tikz/keylengthn}{\len}
        \pgf@x=\len
    }
    \saveddimen\savedlengthne
    {
        \pgfkeysgetvalue{/tikz/keylengthne}{\len}
        \pgf@x=\len
    }
    \saveddimen\savedlengthsw
    {
        \pgfkeysgetvalue{/tikz/keylengthsw}{\len}
        \pgf@x=\len
    }
    \saveddimen\savedlengths
    {
        \pgfkeysgetvalue{/tikz/keylengths}{\len}
        \pgf@x=\len
    }
    \saveddimen\savedlengthse
    {
        \pgfkeysgetvalue{/tikz/keylengthse}{\len}
        \pgf@x=\len
    }
    \saveddimen\overallwidth
    {
        \pgfkeysgetvalue{/tikz/width}{\minwidth}
        \pgf@x=\wd\pgfnodeparttextbox
        \ifdim\pgf@x<\minwidth
            \pgf@x=\minwidth
        \fi
    }
    \savedanchor{\upperrightcorner}
    {
        \pgf@y=.5\ht\pgfnodeparttextbox
        \advance\pgf@y by -.5\dp\pgfnodeparttextbox
        \pgf@x=.5\wd\pgfnodeparttextbox
    }
    \anchor{north}
    {
        \pgf@x=0pt
        \pgf@y=0.5\morphismheight
    }
    \anchor{north east}
    {
        \pgf@x=\overallwidth
        \multiply \pgf@x by 2
        \divide \pgf@x by 5
        \pgf@y=0.5\morphismheight
    }
    \anchor{east}
    {
        \pgf@x=\overallwidth
        \divide \pgf@x by 2
        \advance \pgf@x by 5pt
        \pgf@y=0pt
    }
    \anchor{west}
    {
        \pgf@x=-\overallwidth
        \divide \pgf@x by 2
        \advance \pgf@x by -5pt
        \pgf@y=0pt
    }
    \anchor{north west}
    {
        \pgf@x=-\overallwidth
        \multiply \pgf@x by 2
        \divide \pgf@x by 5
        \pgf@y=0.5\morphismheight
    }
    \anchor{connect nw}
    {
        \pgf@x=-\overallwidth
        \multiply \pgf@x by 2
        \divide \pgf@x by 5
        \pgf@y=0.5\morphismheight
        \advance\pgf@y by \savedlengthnw
    }
    \anchor{connect ne}
    {
        \pgf@x=\overallwidth
        \multiply \pgf@x by 2
        \divide \pgf@x by 5
        \pgf@y=0.5\morphismheight
        \advance\pgf@y by \savedlengthne
    }
    \anchor{connect sw}
    {
        \pgf@x=-\overallwidth
        \multiply \pgf@x by 2
        \divide \pgf@x by 5
        \pgf@y=-0.5\morphismheight
        \advance\pgf@y by -\savedlengthsw
    }
    \anchor{connect se}
    {
        \pgf@x=\overallwidth
        \multiply \pgf@x by 2
        \divide \pgf@x by 5
        \pgf@y=-0.5\morphismheight
        \advance\pgf@y by -\savedlengthse
    }
    \anchor{connect n}
    {
        \pgf@x=0pt
        \pgf@y=0.5\morphismheight
        \advance\pgf@y by \savedlengthn
    }
    \anchor{connect s}
    {
        \pgf@x=0pt
        \pgf@y=-0.5\morphismheight
        \advance\pgf@y by -\savedlengths
    }
    \anchor{south east}
    {
        \pgf@x=\overallwidth
        \multiply \pgf@x by 2
        \divide \pgf@x by 5
        \pgf@y=-0.5\morphismheight
    }
    \anchor{south west}
    {
        \pgf@x=-\overallwidth
        \multiply \pgf@x by 2
        \divide \pgf@x by 5
        \pgf@y=-0.5\morphismheight
    }
    \anchor{south}
    {
        \pgf@x=0pt
        \pgf@y=-0.5\morphismheight
    }
    \anchor{text}
    {
        \upperrightcorner
        \pgf@x=-\pgf@x
        \pgf@y=-\pgf@y
    }
    \backgroundpath
    {
        \pgfsetstrokecolor{black}
        \pgfsetlinewidth{\thickness}
        \begin{scope}
                \ifhflip
                    \pgftransformyscale{-1}
                \fi
                \ifvflip
                    \pgftransformxscale{-1}
                \fi
                \ifhvflip
                    \pgftransformxscale{-1}
                    \pgftransformyscale{-1}
                \fi
                \pgfpathmoveto{\pgfpoint
                    {-0.5*\overallwidth-5pt}
                    {0.5*\morphismheight}}
                \pgfpathlineto{\pgfpoint
                    {0.5*\overallwidth+5pt}
                    {0.5*\morphismheight}}
                \ifwedge
                    \pgfpathlineto{\pgfpoint
                        {0.5*\overallwidth + 15pt}
                        {-0.5*\morphismheight}}
                \else
                    \pgfpathlineto{\pgfpoint
                        {0.5*\overallwidth + 5pt}
                        {-0.5*\morphismheight}}
                \fi
                \pgfpathlineto{\pgfpoint
                    {-0.5*\overallwidth-5pt}
                    {-0.5*\morphismheight}}
                \pgfpathclose
                \pgfusepath{stroke}
        \end{scope}
        \ifconnectnw
            \pgfpathmoveto{\pgfpoint
                {-0.4*\overallwidth}
                {0.5*\morphismheight}}
            \pgfpathlineto{\pgfpoint
                {-0.4*\overallwidth}
                {0.5*\morphismheight+\savedlengthnw}}
            \pgfusepath{stroke}
        \fi
        \ifconnectne
            \pgfpathmoveto{\pgfpoint
                {0.4*\overallwidth}
                {0.5*\morphismheight}}
            \pgfpathlineto{\pgfpoint
                {0.4*\overallwidth}
                {0.5*\morphismheight+\savedlengthne}}
            \pgfusepath{stroke}
        \fi
        \ifconnectsw
            \pgfpathmoveto{\pgfpoint
                {-0.4*\overallwidth}
                {-0.5*\morphismheight}}
            \pgfpathlineto{\pgfpoint
                {-0.4*\overallwidth}
                {-0.5*\morphismheight-\savedlengthsw}}
            \pgfusepath{stroke}
        \fi
        \ifconnectse
            \pgfpathmoveto{\pgfpoint
                {0.4*\overallwidth}
                {-0.5*\morphismheight}}
            \pgfpathlineto{\pgfpoint
                {0.4*\overallwidth}
                {-0.5*\morphismheight-\savedlengthse}}
            \pgfusepath{stroke}
        \fi
        \ifconnectn
            \pgfpathmoveto{\pgfpoint
                {0pt}
                {0.5*\morphismheight}}
            \pgfpathlineto{\pgfpoint
                {0pt}
                {0.5*\morphismheight+\savedlengthn}}
            \pgfusepath{stroke}
        \fi
        \ifconnects
            \pgfpathmoveto{\pgfpoint
                {0pt}
                {-0.5*\morphismheight}}
            \pgfpathlineto{\pgfpoint
                {0pt}
                {-0.5*\morphismheight-\savedlengths}}
            \pgfusepath{stroke}
        \fi
        \ifconnectnwf
            \draw [forward arrow style] (-0.4*\overallwidth,0.5*\morphismheight)
                to (-0.4*\overallwidth,0.5*\morphismheight+\savedlengthnw);
        \fi
        \ifconnectnef
            \draw [forward arrow style] (0.4*\overallwidth,0.5*\morphismheight)
                to (0.4*\overallwidth,0.5*\morphismheight+\savedlengthne);
        \fi
        \ifconnectswf
            \draw [forward arrow style] (-0.4*\overallwidth,-0.5*\morphismheight-\savedlengthsw)
                to (-0.4*\overallwidth,-0.5*\morphismheight);
        \fi
        \ifconnectsef
            \draw [forward arrow style] (0.4*\overallwidth,-0.5*\morphismheight-\savedlengthse)
                to (0.4*\overallwidth,-0.5*\morphismheight);
        \fi
        \ifconnectnf
            \draw [forward arrow style] (0,0.5*\morphismheight)
                to (0,0.5*\morphismheight+\savedlengthn);
        \fi
        \ifconnectsf
            \draw [forward arrow style] (0,-0.5*\morphismheight-\savedlengths)
                to (0,-0.5*\morphismheight);
        \fi
        \ifconnectnwr
            \draw [reverse arrow style] (-0.4*\overallwidth,0.5*\morphismheight)
                to (-0.4*\overallwidth,0.5*\morphismheight+\savedlengthnw);
        \fi
        \ifconnectner
            \draw [reverse arrow style] (0.4*\overallwidth,0.5*\morphismheight)
                to (0.4*\overallwidth,0.5*\morphismheight+\savedlengthne);
        \fi
        \ifconnectswr
            \draw [reverse arrow style] (-0.4*\overallwidth,-0.5*\morphismheight-\savedlengthsw)
                to (-0.4*\overallwidth,-0.5*\morphismheight);
        \fi
        \ifconnectser
            \draw [reverse arrow style] (0.4*\overallwidth,-0.5*\morphismheight-\savedlengthse)
                to (0.4*\overallwidth,-0.5*\morphismheight);
        \fi
        \ifconnectnr
            \draw [reverse arrow style] (0,0.5*\morphismheight)
                to (0,0.5*\morphismheight+\savedlengthn);
        \fi
        \ifconnectsr
            \draw [reverse arrow style] (0,-0.5*\morphismheight-\savedlengths)
                to (0,-0.5*\morphismheight);
        \fi
    }
}
\pgfdeclareshape{swish right}
{
    \savedanchor\centerpoint
    {
        \pgf@x=0pt
        \pgf@y=0pt
    }
    \anchor{center}{\centerpoint}
    \anchorborder{\centerpoint}
    \anchor{north}
    {
        \pgf@x=\minimummorphismwidth
        \divide\pgf@x by 5
        \pgf@y=\morphismheight
        \divide\pgf@y by 2
        \advance\pgf@y by \connectheight
    }
    \anchor{south}
    {
        \pgf@x=-\minimummorphismwidth
        \divide\pgf@x by 5
        \pgf@y=-\morphismheight
        \divide\pgf@y by 2
        \advance\pgf@y by -\connectheight
    }
    \backgroundpath
    {
        \pgfsetstrokecolor{black}
        \pgfsetlinewidth{\thickness}
        \pgfpathmoveto{\pgfpoint
            {-0.2*\minimummorphismwidth}
            {-0.5*\morphismheight-\connectheight}}
        \pgfpathcurveto
            {\pgfpoint{-0.2*\minimummorphismwidth}{0pt}}
            {\pgfpoint{0.2*\minimummorphismwidth}{0pt}}
            {\pgfpoint
                {0.2*\minimummorphismwidth}
                {0.5*\morphismheight+\connectheight}}
        \pgfusepath{stroke}
    }
}
\pgfdeclareshape{swish left}
{
    \savedanchor\centerpoint
    {
        \pgf@x=0pt
        \pgf@y=0pt
    }
    \anchor{center}{\centerpoint}
    \anchorborder{\centerpoint}
    \anchor{north}
    {
        \pgf@x=-\minimummorphismwidth
        \divide\pgf@x by 5
        \pgf@y=\morphismheight
        \divide\pgf@y by 2
        \advance\pgf@y by \connectheight
    }
    \anchor{south}
    {
        \pgf@x=\minimummorphismwidth
        \divide\pgf@x by 5
        \pgf@y=-\morphismheight
        \divide\pgf@y by 2
        \advance\pgf@y by -\connectheight
    }
    \backgroundpath
    {
        \pgfsetstrokecolor{black}
        \pgfsetlinewidth{\thickness}
        \pgfpathmoveto{\pgfpoint
            {0.2*\minimummorphismwidth}
            {-0.5*\morphismheight-\connectheight}}
        \pgfpathcurveto
            {\pgfpoint{0.2*\minimummorphismwidth}{0pt}}
            {\pgfpoint{-0.2*\minimummorphismwidth}{0pt}}
            {\pgfpoint
                {-0.2*\minimummorphismwidth}
                {0.5*\morphismheight+\connectheight}}
        \pgfusepath{stroke}
    }
}
\pgfdeclareshape{state}
{
    \savedanchor\centerpoint
    {
        \pgf@x=0pt
        \pgf@y=0pt
    }
    \anchor{center}{\centerpoint}
    \anchorborder{\centerpoint}
    \saveddimen\overallwidth
    {
        \pgf@x=3\wd\pgfnodeparttextbox
        \ifdim\pgf@x<\minimumstatewidth
            \pgf@x=\minimumstatewidth
        \fi
    }
    \savedanchor{\upperrightcorner}
    {
        \pgf@x=.5\wd\pgfnodeparttextbox
        \pgf@y=.5\ht\pgfnodeparttextbox
        \advance\pgf@y by -.5\dp\pgfnodeparttextbox
    }
    \anchor{A}
    {
        \pgf@x=-\overallwidth
        \divide\pgf@x by 4
        \pgf@y=0pt
    }
    \anchor{B}
    {
        \pgf@x=\overallwidth
        \divide\pgf@x by 4
        \pgf@y=0pt
    }
    \anchor{text}
    {
        \upperrightcorner
        \pgf@x=-\pgf@x
        \ifhflip
            \pgf@y=-\pgf@y
            \advance\pgf@y by 0.4\stateheight
        \else
            \pgf@y=-\pgf@y
            \advance\pgf@y by -0.4\stateheight
        \fi
    }
    \backgroundpath
    {
        \pgfsetstrokecolor{black}
        \pgfsetlinewidth{\thickness}
        \pgfpathmoveto{\pgfpoint{-0.5*\overallwidth}{0}}
        \pgfpathlineto{\pgfpoint{0.5*\overallwidth}{0}}
        \ifhflip
            \pgfpathlineto{\pgfpoint{0}{\stateheight}}
        \else
            \pgfpathlineto{\pgfpoint{0}{-\stateheight}}
        \fi
        \pgfpathclose
        \ifblack
            \pgfsetfillcolor{black!35}
            \pgfusepath{fill,stroke}
        \else
            \pgfusepath{stroke}
        \fi
    }
}
\makeatother

\newcommand{\tinymult}[1][dot]{
\smash{\raisebox{-2pt}{\hspace{-5pt}\ensuremath{\begin{pic}[scale=0.4,string,yscale=-1]
    \node (0) at (0,0) {};
    \node[#1, inner sep=1.5pt] (1) at (0,0.55) {};
    \node (2) at (-0.5,1) {};
    \node (3) at (0.5,1) {};
    \draw (0.center) to (1.center);
    \draw (1.center) to [out=left, in=down, out looseness=1.5] (2.center);
    \draw (1.center) to [out=right, in=down, out looseness=1.5] (3.center);
\end{pic}
}\hspace{-3pt}}}}

\newcommand{\tinyleftdivide}[1][dot]{
\smash{\raisebox{-2pt}{\hspace{-5pt}\ensuremath{\begin{pic}[scale=0.4,string]
    \node (0) at (0,0) {};
    \node[#1,scale=0.9, inner sep=1.5pt] (1) at (0,0.55) {};
    \node[blackdot,scale=0.5] (2) at (-0.7,1) {};
    \node (4) at (-1.3,0){};
    \node (3) at (0.5,1.2) {};
    \draw (0.center) to (1.center);
    \draw (1.center) to [out=left, in=right, out looseness=1] (2.center);
    \draw (1.center) to [out=right, in=down, out looseness=1.5] (3.center);
    \draw[out=up,in=left,looseness=1.5] (4.center) to (2.center);
\end{pic}
}\hspace{-1pt}}}}
\newcommand{\tinyspider}[1][dot]{
\smash{\raisebox{-2pt}{\hspace{-5pt}\ensuremath{\begin{pic}[scale=0.4,string,yscale=1]
    \node (0) at (0,0) {};
    \node[#1, inner sep=1.5pt] (1) at (0,0.55) {};
    \node (2) at (-0.5,1) {};
    \node (3) at (0.5,1) {};
    \node (4) at (-0.5,0) {};
    \node (5) at (0.5,0) {};
    \draw (4.center) to [out=up, in=left, out looseness=1.5](1.center);
    \draw (5.center) to [out=up, in=right, out looseness=1.5](1.center);
    \draw (1.center) to [out=left, in=down, out looseness=1.5] (2.center);
    \draw (1.center) to [out=right, in=down, out looseness=1.5] (3.center);
    \draw[densely dotted] (-0.365,0.1) to (0.5,0.1){};     
    \draw[densely dotted] (-0.365,0.9) to (0.5,0.9){};
\end{pic}
}\hspace{-3pt}}}}

\newcommand{\tinyhandle}[1][dot]{\raisebox{-2pt}{\ensuremath{\hspace{-3pt}\begin{pic}[scale=0.4,string]
        \node (0) at (0,0) {};
        \node[dot, inner sep=1.0pt] (1) at (0,0.3) {};
        \node[dot, inner sep=1.0pt] (2) at (0,0.7) {};
        \node (3) at (0,1) {};
        \draw (0.center) to (1.center);
        \draw (2.center) to (3.center);
        \draw[in=180, out=180, looseness=2] (1.center) to (2.center);
        \draw[in=0, out=0, looseness=2] (1.center) to (2.center);
\end{pic}\hspace{-1pt}}}}

\renewcommand\dag{\ensuremath{\dagger}}